\documentclass[12pt]{article}
\pdfoutput=1
\usepackage{makeidx}
\makeindex
\usepackage[a4paper]{geometry}
\usepackage{jheppub,amsmath,amssymb,amsfonts,amsxtra, mathrsfs,graphics,graphicx,amsthm,epsfig,ytableau,bm,longtable,float,color,tikz,mathtools,xfrac,footnote,rotating,lscape}
\usetikzlibrary{decorations.pathmorphing}
\usetikzlibrary{decorations.markings}
\usetikzlibrary{arrows, decorations.markings, calc, fadings, decorations.pathreplacing, patterns, decorations.pathmorphing, positioning}
\usepackage{tikz-cd}

\usetikzlibrary{positioning,shapes}
\usetikzlibrary{chains}
\usetikzlibrary{arrows,fit,decorations.pathreplacing}
\tikzstyle{every picture}+=[remember picture]
\tikzstyle{na} = [baseline=-.5ex]

\usepackage{empheq}
\usepackage{multirow}
\usepackage{booktabs}
\usepackage[american]{babel}

\usepackage[latin1]{inputenc}

\makeatletter
\newcommand{\vast}{\bBigg@{1}}
\newcommand{\Vast}{\bBigg@{5}}
\makeatother

\usepackage{hyperref}

\numberwithin{equation}{section}

\newcommand{\ie}{\textit{i.e.}}

\numberwithin{equation}{section}

\newcommand{\nn}{\nonumber}
\newcommand{\mat}[1]{\begin{pmatrix} #1 \end{pmatrix}}

\newcommand{\be}{\begin{equation}} \newcommand{\ee}{\end{equation}}
\newcommand{\bea}{\begin{equation} \begin{aligned}} \newcommand{\eea}{\end{aligned} \end{equation}}

\def\U{\mathrm{U}}

\newcommand{\wb}{\overline}

\newcommand{\smallstrut}{\rule{0pt}{.6em}}

\DeclareMathOperator{\Tr}{Tr}

\DeclareMathOperator{\sign}{sign}

\DeclareMathOperator{\re}{\mathbb{R}e}

\DeclareMathOperator{\Li}{Li}

\newcommand{\cC}{\mathcal{C}}

\newcommand{\cG}{\mathcal{G}}

\newcommand{\cI}{\mathcal{I}}

\newcommand{\cN}{\mathcal{N}}
\newcommand{\cO}{\mathcal{O}}

\newcommand{\cV}{\mathcal{V}}

\newcommand{\bB}{\mathbb{B}}

\newcommand{\bZ}{\mathbb{Z}}

\newcommand{\fh}{\mathfrak{h}}
\newcommand{\fm}{\mathfrak{m}}

\newcommand{\fn}{\mathfrak{n}}
\newcommand{\fp}{\mathfrak{p}}

\newcommand{\fR}{\mathfrak{R}}

\newcommand{\ft}{\mathfrak{t}}

\usepackage[Symbol]{upgreek}
\usepackage{bm}

\usepackage{calligra}
\DeclareMathAlphabet{\mathcalligra}{T1}{calligra}{m}{n}

\newtheorem{theorem}{Theorem}

\title{ Large $N$ matrix models for 3d  ${\cal N}=2$ theories: twisted index, free energy and black holes 
}

\author{Seyed Morteza Hosseini}
\author{and Alberto Zaffaroni}

\affiliation{Dipartimento di Fisica, Universit\`a di Milano-Bicocca, I-20126 Milano, Italy}
\affiliation{INFN, sezione di Milano-Bicocca, I-20126 Milano, Italy}

\emailAdd{Morteza.Hosseini@mib.infn.it}
\emailAdd{Alberto.Zaffaroni@mib.infn.it}

\abstract{We provide general formulae for the topologically twisted index  of a general three-dimensional  $\cN \geq 2$ gauge theory
with an M-theory or massive type IIA dual in the large $N$ limit. The index is defined as the supersymmetric path integral of the theory on $S^2\times S^1$ in the presence of background magnetic fluxes for the R- and global symmetries and it is conjectured  to reproduce the entropy of magnetically charged static BPS AdS$_4$ black holes.  For a class of theories with an M-theory dual,  we show that the logarithm of the index scales indeed as $N^{3/2}$ (and $N^{5/3}$ in the massive type IIA case). We find an intriguing relation with the  (apparently unrelated)  large $N$ limit of the partition function on $S^3$. We also provide a universal formula for extracting the index from the large $N$ partition function on $S^3$ and its derivatives and point out its analogy with the attractor mechanism for AdS black holes.  
}

\begin{document}

\setcounter{tocdepth}{2}
\maketitle

%
%

\date{Dated: \today}




\section{Introduction}

In this paper we study the large $N$ behavior of the topologically twisted index introduced in \cite{Benini:2015noa}  for  three-dimensional  $\cN \geq 2$ gauge theories.
It is defined as the partition function of the theory on $S^2 \times S^1$ with a topological twist
 along $S^2$ \cite{Witten:1988hf,Witten:1991zz} and it is  a function of magnetic charges and chemical potentials for the flavor symmetries. 
The large $N$ limit of the index for the ABJM theory was successfully used in  \cite{Benini:2015eyy} to provide the first microscopic counting of the microstates
of an AdS$_4$ black hole. Here we extend the analysis of  \cite{Benini:2015eyy} to a larger class of $\cN \geq 2$ theories with an M-theory or massive type IIA dual, containing bi-fundamental,
adjoint and (anti-)fundamental chiral matter.  Most of the theories proposed in the literature  are obtained by adding Chen-Simons terms \cite{Aharony:2008ug,Hanany:2008cd,Hanany:2008fj,Martelli:2008si,Franco:2008um,Davey:2009sr,Hanany:2009vx,Franco:2009sp}  or by flavoring  \cite{Gaiotto:2009tk,Jafferis:2009th,Benini:2009qs}  four-dimensional quivers  describing D3-branes probing CY$_3$ singularities. We refer to these theories as having a four-dimensional parent. They all have an M-theory phase where the index is expected to scale as $N^{3/2}$.  The main motivation for studying the  large $N$ limit of the index for these theories comes indeed form the attempt to extend the result of \cite{Benini:2015eyy} to a larger class of black holes, and we hope to report on the subject soon. However,  the matrix model computing the index reveals an interesting structure at large $N$ which deserves attention by itself. In particular, we will point out analogies and relations with other matrix models appeared in the
literature on three-dimensional  $\cN \geq 2$ gauge theories.



The index can be evaluated using supersymmetric  localization and it reduces to a matrix model. It can be written as the contour integral 
\begin{equation}  Z (\fn, y) = \frac1{|W|} \; \sum_{\fm \,\in\, \Gamma_\fh} \; \oint_\cC Z_{\text{int}} (\fm, x;  \fn , y) \end{equation}
of a meromorphic differential form in variables $x$ parameterizing the Cartan subgroup and subalgebra of the gauge group, summed over
the lattice of magnetic charges $\fm$ of the group. The index depends on  complex fugacities $y$ and magnetic charges $\fn$ for the flavor symmetries.
As a difference with other well known matrix models arising from supersymmetric localization in three dimensions, like the partition function on $S^3$  \cite{Kapustin:2009kz,Jafferis:2010un,Hama:2010av} or the superconformal index \cite{Kim:2009wb}, in the large $N$ limit all the gauge magnetic fluxes contribute to the
integral making difficult its evaluation. Here we use the strategy employed in  \cite{Benini:2015eyy} to explicitly resum the integrand and consider
the contour integral of the sum
\begin{equation} Z_{\text{resummed}}  ( x;  \fn , y) = \frac1{|W|} \; \sum_{\fm \,\in\, \Gamma_\fh}  Z_{\text{int}} (\fm, x;  \fn , y) \end{equation}
which is a complicated rational function of $x$. We set up, as in  \cite{Benini:2015eyy}, an auxiliary large $N$ problem devoted to find
the positions  of the  poles of $Z_{\text{resummed}}$ in the plane $x$. We write a set of algebraic equations for the position of the poles,
which we call {\it Bethe Ansatz Equations} (BAEs), and we write a {\it Bethe functional} whose derivative reproduces the  BAEs.
The method for solving the BAEs is similar to that used in \cite{Herzog:2010hf,Jafferis:2011zi} for the large $N$ limit of the partition function on $S^3$
in the M-theory limit and the one used for the partition function on $S^5$ of five-dimensional theories \cite{Jafferis:2012iv,Minahan:2013jwa,Minahan:2014hwa}. We take
an ansatz for the eigenvalues where the imaginary parts grow in the large $N$ limit  as some power of $N$. The solution of the BAEs in the large $N$ limit is then used to evaluate $Z (\fn, y)$ using the residue theorem.  
In this last step we need to take into account (exponentially small) corrections to the large $N$ limit of the BAEs which contribute to the index due
to the singular logarithmic behavior of its integrand.

We focus on the limit where $N$ is much greater than the Chern-Simons couplings $k_a$. For the class of  quivers we are considering, this limit corresponds  to an M-theory description when  $\sum_a k_a=0$ and a massive type IIA one when $\sum_a k_a\neq 0$. We recover the known scalings $N^{3/2}$ and $N^{5/3}$ 
for the M-theory and massive type IIA phase, respectively. Similarly to   \cite{Jafferis:2011zi}, we find that, in order
to have a consistent $N^{3/2}$ scaling of the index in the M-theory phase, we need to impose some constraints on the quiver. In particular, quivers with a chiral 4d parent
are not allowed, as in  \cite{Jafferis:2011zi}. They are instead allowed in the massive type IIA phase.

In the course of our analysis, we find a number of  interesting general results.

First, we find a simple universal  formula for computing the index from the Bethe potential, $\wb{\mathcal{V}}(\Delta_I)$,  as a function of the chemical potentials,
\begin{equation}\label{Z large N conjecture0}
 {\quad \rule[-1.4em]{0pt}{3.4em}
 \re\log Z = - \frac{2}{\pi} \, \wb{\mathcal{V}}(\Delta_I) \,
 - \sum_{I}\, \left[ \left(\fn_I - \frac{\Delta_I}{\pi}\right) \frac{\partial \wb{\mathcal{V}}(\Delta_I)}{\partial \Delta_I}
  \right] \, .
 \quad}
\end{equation}
We call this the {\it index theorem}.
It allows to avoid the many technicalities involved in taking the residues and including exponentially small corrections to the index. 
By comparing the index theorem with the attractor formula for the entropy of asymptotically AdS$_4$ black holes,  we are also led to conjecture a relation between
the Bethe potential and the prepotential of the dimensionally truncated  gauged supergravity describing the compactification on AdS$_4\times Y_7$, with $Y_7$  a Sasaki-Einstein manifold. 
This relation is discussed in Section \ref{discussion}.
 
 Secondly,  we find an explicit relation between the Bethe potential and the $S^3$ free energy of the same $\cN \geq 2$ gauge theory.
 Although the two matrix models are quite different at finite $N$,  the BAEs and the functional form of the Bethe potential
 in the large $N$ limit are {\it identical} to the matrix model equations of motion and  free energy functional for the path integral on $S^3$ found in   \cite{Jafferis:2011zi}. 
 This result implies that the index can be  extracted from the free energy on $S^3$ and its derivatives in the large $N$ limit. It also implies a relation with the volume
 functional of (Sasakian deformations of) the internal manifold $Y_7$.  
 These relations deserve a better understanding.
 
 In this paper we give the general rules for constructing the Bethe potential and the index for a generic Yang-Mills-Chen-Simons theory with bi-fundamental, adjoint and fundamental fields and few explicit examples of their application. Many other examples can be found in an upcoming paper by one of the authors \cite{Hosseini:2016ume},
including models for well-known homogeneous Sasaki-Einstein manifolds, $N^{0,1,0},Q^{1,1,1},V^{5,2}$,  and various nontrivial checks of dualities.  We leave the most interesting part
of the story concerning applications to the microscopic counting for AdS$_4$ black holes for the future. 

The paper is organized as follows. In Section \ref{index} we review the definition of the topologically twisted index and the strategy for determining its large $N$ limit used 
in \cite{Benini:2015eyy}. In Section \ref{proof vanishing forces} we give the general rules for constructing the Bethe potential and the index for a generic Yang-Mills-Chen-Simons theory with bi-fundamentals, adjoint and fundamental fields with $N^{3/2}$ scaling. In Section \ref{sec:freeenergy} we prove the identity of
the Bethe potential and the $S^3$ free energy at large $N$. In Section \ref{sec:index} we derive the index theorem that allows to express the index at large $N$  in terms of
the Bethe potential and its derivatives. In Section \ref{N53} we discuss the rules  for a $N^{5/3}$ scaling.  In Section \ref{discussion} we give a discussion of some open issues and point out analogies with the attractor mechanism for AdS black holes. Appendices
\ref{general rules N32} and \ref{general rules N53} contains the explicit derivations of the rules for $N^{3/2}$ and $N^{5/3}$ scalings, respectively. Appendix \ref{SPP} contains
an explicit example, based on the quiver for the Suspended Pinch Point.

\section{The  topologically twisted index}
\label{index}
 The topologically twisted index  of an $\cN \geq 2$ gauge theory in three dimensions is defined as the partition function on $S^2 \times S^1$ with a topological twist
 along $S^2$ \cite{Benini:2015noa}. It depends on a choice of fugacities $y$ for the global symmetries and magnetic charges $\fn$ on $S^2$ parameterizing the twist. The index can be computed using localization  and it is given by a matrix integral over the zero-mode gauge variables and it is summed over a lattice of gauge magnetic charges on $S^2$.
Explicitly, for a theory with gauge group $G$ of rank $r$ and a set of chiral multiplets  transforming in representations $\fR_I$ of $G$, the index is given by \cite{Benini:2015noa}
\be
\label{path}
Z (\fn, y) = \frac1{|W|} \; \sum_{\fm \,\in\, \Gamma_\fh} \; \oint_\cC \;   \prod_{\text{Cartan}} \left (\frac{dx}{2\pi i x}  x^{k \fm} \right ) \prod_{\alpha \in G} (1-x^\alpha) \,  \prod_I \prod_{\rho_I \in \fR_I} \bigg( \frac{x^{\rho_I/2} \, y_I^{1/2}}{1-x^{\rho_I} \, y_I} \bigg)^{\rho_I(\fm)- \fn_I  +1}  \, ,
\ee
where $\alpha$ are the roots of $G$ and $\rho_I$ are the  weights of the representation $\fR_I$.
In this formula\footnote{$\beta$ is the radius of $S^1$.}, $x=e^{i(A_t + i\beta \sigma)}$ parameterizes the gauge zero modes, where $A_t$ is a Wilson line on $S^1$ and runs over the maximal torus of $G$ while $\sigma$ is the real scalar in  the vector multiplet and runs over the corresponding Cartan subalgebra. $\fm$ are  gauge magnetic fluxes  living in the co-root lattice $\Gamma_\fh$ of $G$ (up to gauge transformations).  The index is integrated over $x$ and summed over $\fm$.  $k$ is the Chern-Simons coupling for the group $G$, and there can be a different one for
each Abelian and simple factor in $G$. Supersymmetric localization selects a particular contour of integration and the final result can be formulated in terms of the Jeffrey-Kirwan residue \cite{Benini:2015noa}.

The index depends on a choice of fugacities $y_I$  for the flavor group and a choice 
of integer magnetic charges  $\fn_I$  for the R-symmetry of the theory. In an $\cN \geq 2$ theory, the R-symmetry can mix with the global symmetries
and we can also write
\be
\fn_I = q_I + \fp_I \, ,
\ee
where $q_I$ is a reference R-symmetry and $\fp_I$ magnetic charges under the flavor symmetries of the theory. Both $y_I$ and $\fn_I$ are thus parameterized by the global symmetries of the theory. Each monomial term $W$ in the superpotential imposes a constraint
\be
\prod_{I\in W} y_{I} =1 \, , \qquad\qquad \sum_{I\in W} \fn_I = 2 \, ,
\ee   
where the product and sum are restricted to the fields entering in $W$. Each Abelian gauge group in three dimensions is associated with a topological $\U(1)$ symmetry.
The contribution of a  topological symmetry with fugacity $\xi =e^{i z}$ and magnetic flux $\ft$ to the index is given by
\be
\label{topological}
Z^\text{top} = x^\ft \, \xi^\fm \, ,
\ee
where $x$ is the gauge variable of the corresponding $\U(1)$ gauge field.

In this paper we are interested in  the large $N$ limit of the topologically twisted index  for theories with unitary gauge groups and matter transforming in the fundamentals, bi-fundamentals and adjoint representation. As in \cite{Benini:2015eyy},  we evaluate the matrix model in two steps.
We first perform the summation over magnetic fluxes introducing a large cut-off $M$.\footnote{According to the rules in  \cite{Benini:2015noa}, the residues to take in \eqref{path} depend on the
sign of the Chern-Simons couplings. We can choose a set of co-vectors in the Jeffrey-Kirwan prescription such that the contribution comes from  residues with $\fm_a\leq 0$ for $k_a > 0$, residues with $\fm_a\geq 0$ for $k_a < 0$  and residues in the origin. We can then take a large positive integer $M$ and perform the summations in Eq.\,\eqref{path}, with $\fm_a \leq M-1 \, (k_a > 0)$ and $\fm_a \geq 1-M \, (k_a < 0)$.}
The result of this summation produces terms in the integrand of the form 
\begin{equation}
 \prod_{i=1}^{N} \frac{\left(e^{i B_i^{(a)}}\right)^M}{e^{i B_i^{(a)}} - 1} \, ,
\end{equation}
where we defined
\begin{align}\label{BA expression}
 e^{i \sign(k_a) B_i^{(a)}} & = \xi^{(a)} (x_i^{(a)})^{k_a}
 \prod_{\substack{\text{bi-fundamentals} \\ (a,b) \text{ and } (b,a) }} \prod_{j=1}^N \frac{\sqrt{\frac{x_i^{(a)}}{x_j^{(b)}} y_{(a,b)}}}{1 - \frac{x_i^{(a)}}{x_j^{(b)}} y_{(a,b)}}
 \frac{1 - \frac{x_j^{(b)}}{x_i^{(a)}} y_{(b,a)}}{\sqrt{\frac{x_j^{(b)}}{x_i^{(a)}} y_{(b,a)}}} \nn \\
 & \hskip 1truecm \times \prod_{\substack{\text{fundamentals} \\ a}} \frac{\sqrt{x_i^{(a)} y_a}}{1 - x_i^{(a)} y_a} \,\, 
 \prod_{\substack{\text{anti-fundamentals} \\ a}} \frac{1 - \frac{1}{x_i^{(a)}} \tilde y_a}{\sqrt{\frac{1}{x_i^{(a)}} \tilde y_a}} \, ,
\end{align}
and adjoints are identified with bi-fundamentals connecting the same gauge group ($a=b$).
In this way the contributions from the residues at the origin have been moved to the solutions of the ``Bethe Ansatz Equations'' (BAEs)
\begin{align}\label{BAEs}
 e^{i \sign(k_a) B_i^{(a)}}=1 \, .
\end{align}
It is convenient to use the  variables $u_i^{(a)}$ and $\Delta_I$, defined modulo $2\pi$,\footnote{Notice that the index is a holomorphic function of $y_I$ and $\xi$. There is no loss of generality in  restrict  to  the case of purely imaginary chemical potentials $\Delta$ in \eqref{u Delta variable}.}
\be\label{u Delta variable}
x_i^{(a)}= e^{i u_i^{(a)}} \;, \qquad\qquad y_I = e^{i\Delta_I}\;, \qquad\qquad \xi^{(a)} = e^{i\Delta_m^{(a)}} \, ,
\ee
and take the logarithm of the Bethe ansatz equations 
\begin{align}\label{log BAEs0}
 0 & = \log \left[ \text{RHS of \eqref{BA expression}} \right] - 2 \pi i n_i^{(a)} 
 \, ,
\end{align}
where $n_i^{(a)}$ are integers that parameterize the angular ambiguities. The BAEs \eqref{log BAEs0} can be obtained as critical points of a ``Bethe potential" $\cV(u_i^{(a)})$.

We then need to solve these auxiliary equations in the large $N$ limit. Once the distribution of poles in the integrand in the large $N$ limit has been found, we can finally evaluate the index by computing the residue of the
resummed integrand of \eqref{path} at the solutions of \eqref{log BAEs0}.  
In the final expression, the dependence on $M$ disappears.

\section{The large $N$ limit of the index}
\label{proof vanishing forces}

We are interested in the properties of the topologically twisted index  in the large $N$ limit of theories with  an M-theory dual. 
We focus on quiver Chern-Simons-Yang-Mills gauge theories with gauge group
\be \label{quiver} \cG = \prod_{a=1}^{|G|} \U(N)_a \, ,\ee
and bi-fundamental, adjoint and fundamental chiral multiplets.  Most of the conjectured theories living on M2-branes probing $\text{CY}_4$ singularities
are of this form. Moreover, many of them are obtained by adding Chern-Simons terms and fundamental flavors to quivers appeared in the four-dimensional literature as 
describing D3-branes probing $\text{CY}_3$ singularities.  We refer to these theories as quivers {\it with a 4d parent}. In order to have a $\text{CY}_4$ moduli space, the  Chern-Simons couplings must satisfy
\be\label{CScontr}
\sum_{a=1}^{|G|} k_a = 0 \, .
\ee
The M-theory phase of these theories is obtained for $N\gg k_a$ and this is the limit we consider here. We expect the index to scale as $N^{3/2}$.

As in \cite{Benini:2015eyy}, we consider the following ansatz for  the large $N$ saddle-point eigenvalue distribution:
\be\label{ansatz alpha}
u_i^{(a)} = i N^{1/2} t_i + v_i^{(a)} \, .
\ee
Notice that the imaginary parts of all the  $u_i^{(a)}$ are equal. In the large $N$ limit, we define the continuous functions $t_i = t(i/N)$ and $v_i^{(a)}= v^{(a)}(i/N)$ and  we introduce the density of eigenvalues 
\be \rho(t) = \frac1N\, \frac{di}{dt}\, , \ee
normalized as $ \int dt\, \rho(t) = 1$. 

The large $N$ limit of the Bethe potential is performed in details in Appendix \ref{Bethe potential rules N32}, generalizing the analysis in \cite{Benini:2015eyy}. Here, we report the final result and some of the crucial subtleties.
We need to require the cancellations of long-range forces in the BAEs, as originally observed in a similar context in  \cite{Jafferis:2011zi}, and this imposes some
constraints on the quiver. Once these are satisfied,  the Bethe potential $\cV$ becomes a local functional of $\rho(t)$ and $v_i^{(a)}(t)$ and it scales as  $N^{3/2}$. 
The same constraints guarantee that the index itself scales as $N^{3/2}$.

\subsection{Cancellation of long-range forces}
\label{long-range}

As in \cite{Jafferis:2011zi}, when bi-fundamentals are present, we need to cancel  long-range forces in the BAEs.  These are detected by considering the force exerted by the eigenvalue $u_j^{(b)}$ on the eigenvalue $u_i^{(a)}$ in \eqref{log BAEs0}. They can grow with large powers of $N$ and need to be canceled by imposing constraint on the quiver and matter content if necessary. Since $u_j^{(b)}- u_i^{(a)} \sim \sqrt{N}$ for $i\neq j$, when the long-range forces vanish, the BAE and the Bethe potential get only
contributions from $i\sim j$ and they become local functionals of $\rho(t)$ and $v_i^{(a)}(t)$.

Let us consider the effects of such long-range forces in the Bethe potential $\cV$. A single bi-fundamental field connecting gauge groups $a$ and $b$ contributes terms of the form
\be\label{N3}
\sum_{i<j} \frac{\left(u_i^{(a)} - u_j^{(b)}\right)^2}{4} - \sum_{i<j} \frac{\left(u_j^{(a)} - u_i^{(b)}\right)^2}{4}\, ,
\ee
to the Bethe potential [see Eq.\,\eqref{forces chiral}].
In the large $N$ limit, they are of order $N^{5/2}$.
In order to cancel these terms, we are then forced, as  in \cite{Jafferis:2011zi}, to consider quivers where for each bi-fundamental connecting $a$ and $b$ there is also a bi-fundamental connecting $b$ and $a$. The contribution of the two bi-fundamentals then cancel out [see Eq.\,\eqref{reflection applied 34} and Eq.\,\eqref{reflection applied 12}].

From a pair of  bi-fundamentals, we  get another contribution to the Bethe potential of the form  [see Eq.\,\eqref{bi-fundamentals forces2}]
\begin{equation}\label{N2ang}
- \frac{1}{2} \left[ \left(\Delta_{(a,b)} - \pi \right) + \left( \Delta_{(b,a)} - \pi \right) \right] \sum_{j\neq i}^N  \left(u_i^{(a)} -u_j^{(b)}\right) \text{sign}{(i-j)} \, .
\end{equation}
This term can be canceled by the contribution of the angular ambiguities in \eqref{log BAEs0} to the Bethe potential $\cV$
\begin{equation}
  2 \pi \sum_{i=1}^N   n_i^{(a)} u_i^{(a)}  \, ,
\end{equation}
provided that,\footnote{This is actually true only when $N$ is odd. For even $N$ we are left with a common factor $\pi \sum_{i=1}^N u_i^{(a)}$ which can be
reabsorbed in the definition of $\xi^{(a)}$.}
\begin{equation}\label{no long-range-forces Bethe}
\sum_{I \in a} (\pi  - \Delta_I) \in 2 \pi \bZ \, ,
\end{equation}
where the sum is taken over all bi-fundamental fields with one leg in the node $a$.\footnote{Adjoint fields are supposed to be counted twice.}
Since for any reasonable quiver the number of arrows entering a node is the same as the number of arrow leaving it, this equation is obviously equivalent to $\sum_{I \in a} \Delta_I \in 2 \pi \bZ$ and can be also written as 
\begin{equation}\label{no long-range-forces Bethe0}
\prod_{I \in a}  y_I =1 \, .
\end{equation}
This condition implies that the sum of the charges under all global symmetries of the bi-fundamental fields at each node must vanish. For quivers with a 4d parent, 
this is equivalent to the absence of anomalies for the global symmetries of the 4d theory. Taking the product over all the nodes in a quiver, we also get
\be\label{globalJ}
\Tr J =0 \, ,
\ee
for any global symmetry of the theory, where the trace is taken over all the bi-fundamental fermions.

There are also contributions to  the Bethe potential  of $\cO(N^2)$.  The Chern-Simons terms give indeed
\be
\sum_a k_a  \sum_i^N  \frac{\left(u_i^{(a)}\right)^2}{2} \, .
\ee
However, the $\cO(N^2)$ term cancels out when the condition \eqref{CScontr} is satisfied. Finally, there is a $O(N^2)$ contributions of the fundamental fields  given by
\eqref{fundcontr}. This  vanishes if the {\it total} number of fundamental and anti-fundamental fields in the quiver is the same. 

We turn next to the large $N$ limit of the index. The vector multiplet contributes a term of $\cO(N^{5/2})$ [see Eq.\,\eqref{gauge entropy}]
\begin{equation}\label{indexscalingg}
 i \sum_{i<j}^N \left( u_i^{(a)} - u_j^{(a)} + \pi \right) \, .
\end{equation}
The contribution of $\cO(N^{5/2})$ of a chiral multiplet is  [see Eq.\,\eqref{forces chiral ba}]:
\begin{align}\label{indexscalingb}
 i \sum_{I \in a} \frac{\left(\fn_I - 1\right)}{2} \sum_{i<j}^N \left( u_i^{(a)} - u_j^{(a)} + \pi \right) \, .
\end{align}
To have a cancellation between terms of $\cO(N^{5/2})$ and $\cO(N^2)$ for each node $a$ we must have
\begin{equation}\label{no long-range forces}
 2 + \sum_{I \in a} \left(\fn_I - 1\right) = 0 \, .
\end{equation}
For a quiver with a 4d parent, this condition is equivalent to the absence of anomalies for the R-symmetry.
If we sum over all the nodes we also obtain the following constraint
\begin{equation}\label{globalR}
 {\quad \rule[-1.4em]{0pt}{3.4em} |G| + \sum_{I} \left(\fn_I - 1\right) = 0\, .\quad}
\end{equation}
The above equation is equivalent to $\Tr R = 0$ for any trial R-symmetry, where the trace is taken over all the bi-fundamental fermions and gauginos.

Summarizing, we can have  a $N^{3/2}$ scaling for the index if for each bi-fundamental connecting $a$ and $b$ there is also a bi-fundamental connecting $b$ and $a$,
the total number of fundamental and anti-fundamental fields in the quiver is equal,  Eq.\,\eqref{no long-range-forces Bethe0} and Eq.\,\eqref{no long-range forces}
are fulfilled. All these conditions are automatically satisfied for quivers with a toric vector-like 4d parent and also for other interesting models like \cite{Martelli:2009ga}.
However, they rule out many interesting chiral quivers appeared in the literature on M2-branes. We note a striking analogy with the conditions imposed in \cite{Jafferis:2011zi}. 

\subsection{Bethe potential at large $N$}
\label{large N Bethe potential rules}

In this section we give the general rules for constructing the  Bethe potential for any $\cN \geq 2$ quiver gauge theory which respects the constraints \eqref{no long-range-forces Bethe0} and \eqref{no long-range forces}:

\begin{enumerate}
 \item Each group $a$ with CS level $k_a$ and chemical potential for the topological symmetry $\Delta_m^{(a)}$ contributes the term
 \begin{equation}\label{CS cV rule}
  {\quad \rule[-1.4em]{0pt}{3.4em} - i k_a N^{3/2} \int dt\, \rho(t)\, t\, v_a(t) - i \Delta_m^{(a)} N^{3/2} \int dt\, \rho(t)\, t\, . \quad}
 \end{equation}
 \item A pair of bi-fundamental fields, one with chemical potential $\Delta_{(a,b)}$ and transforming in the $({\bf N},\overline{\bf N})$
 of $\U(N)_a \times \U(N)_b$ and the other with chemical potential $\Delta_{(b,a)}$ and transforming in the $(\overline{\bf N},{\bf N})$ of $\U(N)_a \times \U(N)_b$, contributes
 \begin{equation}\label{Bethe potential bi-fundamental}
  {\quad \rule[-1.4em]{0pt}{3.4em} i N^{3/2} \int dt\, \rho(t)^2 \left[g_+\left(\delta v(t) + \Delta_{(b,a)}\right) - g_-\left(\delta v(t) - \Delta_{(a,b)}\right)\right]\, , \quad}
 \end{equation}
 where $\delta v(t) \equiv v_b (t) - v_a (t)$. Here, we introduced the polynomial functions
 \begin{equation}\label{gp gm}
  g_\pm(u) = \frac{u^3}6 \mp \frac\pi2 u^2 + \frac{\pi^2}3 u \;,\qquad\qquad g_\pm'(u) = \frac{u^2}2 \mp \pi u + \frac{\pi^2}3 \, ,
 \end{equation}
 and we assumed them to be in the range
 \be
\label{inequalities for delta v0}
0 < \delta v + \Delta_{(b,a)} < 2\pi \;,\qquad\qquad\qquad -2\pi < \delta v - \Delta_{(a,b)} < 0 \, ,
\ee
which can be adjusted by choosing a specific determination for the $\Delta$ that  are  defined modulo $2\pi$. We will also assume, and 
this is certainly true if $\delta v$ assumes the value zero, that
\be
\label{general range of Delta_a0}
0 < \Delta_I < 2\pi \, .
\ee
 \item An adjoint field with chemical potential $\Delta_{(a,a)}$, contributes
 \begin{equation}\label{adjoint cV rule}
  {\quad \rule[-1.4em]{0pt}{3.4em} i g_+(\Delta_{(a,a)}) N^{3/2} \int dt\, \rho(t)^2\, . \quad}
 \end{equation}
 \item A field $X_a$ with chemical potential $\Delta_a$ transforming in the fundamental of $\U(N)_a$, contributes
 \begin{equation}\label{fund cV rule}
  {\quad \rule[-1.4em]{0pt}{3.4em} -\frac{i}{2} N^{3/2} \int dt\, \rho(t)\, |t| \Big[v_a(t) + \big( \Delta_a - \pi \big)\Big] \, , \quad}
 \end{equation}
 while an anti-fundamental field with chemical potential $\tilde \Delta_a$ contributes\footnote{We also assume $0<v_a(t)+\Delta_a<2 \pi$ and $0<-v_a(t)+\tilde \Delta_a<2 \pi$.}
 \begin{equation}\label{anti-fund cV rule}
  {\quad \rule[-1.4em]{0pt}{3.4em} \frac{i}{2} N^{3/2} \int dt\, \rho(t)\, |t| \Big[v_a(t) - \big( \tilde \Delta_a - \pi \big)\Big] \, . \quad}
 \end{equation}
\end{enumerate}

Adding all the previous contributions for all gauge groups and matter fields, we get a local functional $\cV(\rho(t) ,v_a(t), \Delta_I)$ that we need to extremize
 with respect to the continuous functions
$\rho(t)$ and $v_a(t)$ with the constraint $\int dt \rho(t) =1$. Equivalently we can introduce a Lagrange multiplier $\mu$ and extremize
\be\label{bethefunctional}
\cV\left(\rho(t) ,v_a(t), \Delta_I\right) - \mu \left (\int dt \rho(t) -1 \right)  \, .
\ee
This gives the large $N$ limit distribution of poles in the index matrix model.  

The solutions of the BAEs have a typical piece-wise structure.
 Eq.\,\eqref{bethefunctional} is the right functional to extremize when the conditions \eqref{inequalities for delta v0} are satisfied.
 This gives a central region where $\rho(t)$ and $v_a(t)$ vary with continuity as functions of $t$. 
When one of the $\delta v(t)$ associated with a pair of bi-fundamental hits the boundaries of the inequalities \eqref{inequalities for delta v0}, 
it remains frozen to a constant value  $\delta v = -\Delta_{(b,a)}$ (mod $2\pi$) or  $\delta v = \Delta_{(b,a)}$  (mod $2\pi$)
for larger (or smaller) values of $t$. This creates ``tail'' regions where one or more $\delta v$ are frozen and the functional  \eqref{bethefunctional}
is extremized with respect to the remaining variables. In the tails, the derivative of  \eqref{bethefunctional} with respect to the frozen variable is not zero
and it is compensated by subleading terms that we omitted.  To be precise, the equations of motion  [see Eq.\,\eqref{talis eom}] includes
subleading terms
 \begin{equation}\label{tails} \frac{\partial \mathcal{V}}{\partial (\delta v)} + i N \rho \left[ \Li_1 \left( e^{i \left( \delta v + \Delta_{(b,a)} \right)} \right)
 - \Li_1 \left( e^{i \left( \delta v - \Delta_{(a,b)} \right)} \right)\right] = 0 \, ,
 \end{equation}
which are negligible except on the tails, where $\delta v$ has exponentially small correction to the large $N$ constant value 
 \begin{equation}\label{tails2}
 \delta v = - \Delta_{(b,a)} + e^{-N^{1/2} Y_{(b,a)}}  \, , \qquad 
 \delta v = \Delta_{(a,b)}   - e^{-N^{1/2} Y_{(a,b)}} \, ,   \qquad  (\text{mod } 2\pi)\, .
 \end{equation}
The quantities $Y$ are determined by equation \eqref{tails} and contribute to the large $N$ limit of the index.

\subsubsection{The ABJM example} As an example, we briefly review here the solution to the BAEs for the ABJM model found in \cite{Benini:2015eyy}.
A more complicated example, for a $\U(N)^3$ quiver is discussed in Appendix \ref{SPP}. The reader can find many other examples in \cite{Hosseini:2016ume}.
ABJM is a Chern-Simons-matter theory with gauge group $\U(N)_k \times \U(N)_{-k}$, with two pairs of  bi-fundamental fields $A_i$ and $B_i$ transforming in the 
representation $({\bf N},\overline{\bf N})$ and $(\overline{\bf N},{\bf N})$ of the gauge group, respectively, and superpotential
\be
W = \Tr \left( A_1B_1A_2B_2 - A_1B_2A_2B_1 \right) \, .
\ee
We assign chemical potentials $\Delta_{1,2}\in [0,2\pi]$ to $A_i$ and  $\Delta_{3,4}\in [0,2\pi]$ to $B_i$. Invariance of the superpotential under the global symmetries
requires that $\sum_i \Delta_i \in 2\pi \mathbb{Z}$ (or equivalently $\prod_{i} y_i = 1$).  Conditions \eqref{no long-range-forces Bethe0} and \eqref{no long-range forces} are then automatically satisfied. The Bethe potential, for $k=1$\footnote{There is a similar solution  for $k>1$ with $\cV\rightarrow \cV \sqrt{k}$. However, we also need
to take into account that, for $k>1$, there are further identifications among the $\Delta_I$ due to discrete $\mathbb{Z}_k$ symmetries of the quiver.}, reads
\be
\cV =  i N^{3/2} \int dt \left\{ t\, \rho(t)\, \delta v(t) + \rho(t)^2 \left[ \, \sum_{a=3,4} g_+ \left( \delta v(t) + \Delta_a \right) - \sum_{a=1,2} g_- \left( \delta v(t) - \Delta_a \right)\right] \right\} \, .
\ee
The solution for $\sum_i \Delta_i = 2 \pi$ and $\Delta_1\leq \Delta_2$, $\Delta_3\leq \Delta_4$  is as follows \cite{Benini:2015eyy}. We have a central region where 
\be
\begin{aligned}
\rho &= \frac{2\pi \mu + t(\Delta_3 \Delta_4 - \Delta_1 \Delta_2)}{(\Delta_1 + \Delta_3)(\Delta_2 + \Delta_3)(\Delta_1 + \Delta_4)(\Delta_2 + \Delta_4)} \\[.5em]
\delta v &= \frac{\mu(\Delta_1 \Delta_2 - \Delta_3 \Delta_4) + t \sum_{a<b<c} \Delta_a \Delta_b \Delta_c }{ 2\pi \mu + t ( \Delta_3 \Delta_4 - \Delta_1 \Delta_2) }
\end{aligned}
\qquad\qquad -\frac{\mu}{\Delta_4}   < t < \frac{\mu}{\Delta_2} \, .
\ee
When $\delta v$ hits $-\Delta_3$ on the left the solution becomes
\be
\rho = \frac{\mu + t\Delta_3}{(\Delta_1 + \Delta_3)(\Delta_2 + \Delta_3)(\Delta_4 - \Delta_3)} \, , \,\,   \delta v = - \Delta_3 \;, \,\, \qquad  -\frac{\mu}{\Delta_3}  < t <-\frac{\mu}{\Delta_4}  \, ,
\ee
with the exponentially small correction  $Y_3 = (- t\Delta_4 -\mu)/(\Delta_4 - \Delta_3)$, while when  $\delta v$ hits $\Delta_1$ on the right the solution becomes
\be
\rho = \frac{\mu - t \Delta_1}{(\Delta_1 + \Delta_3)(\Delta_1 + \Delta_4)(\Delta_2 - \Delta_1)} \, ,\,\,
\delta v = \Delta_1 \;,\,\, \qquad \frac{\mu}{\Delta_2}  < t < \frac{\mu}{\Delta_1}  \, ,
\ee
with $Y_1 =(t\Delta_2 - \mu)/(\Delta_2 - \Delta_1)$.
Finally, the on-shell Bethe potential is
\be
\label{Vsolution sum 2pi -- end}
\cV = \frac {2i}{3} \mu N^{3/2}  = \frac {2i N^{3/2}}{3} \sqrt{ 2 \Delta_1 \Delta_2 \Delta_3 \Delta_4} \, .
\ee
There is also a solution for $\sum_i \Delta_i = 6 \pi$ which, however, is obtained by the previous one by a discrete symmetry $\Delta_i \rightarrow 2 \pi -\Delta_i$ $\left(y_i\rightarrow y_i^{-1}\right)$.

\subsection{The index at large $N$}

We now turn to the large $N$ limit of the  index for an $\cN \geq 2$ quiver gauge theory without long-range forces. 
Here, we give the rules for constructing the index once we know the large $N$ solution $\rho(t),v_a(t)$ of the BAE, which is obtained by extremizing \eqref{bethefunctional}.
The final result scales as $N^{3/2}$.

\begin{enumerate}
 \item For each group $a$, the contribution of the Vandermonde determinant is
 \begin{equation}\label{gaugecontribution}
  {\quad \rule[-1.4em]{0pt}{3.4em} -\frac{\pi^2}{3} N^{3/2} \int dt\, \rho(t)^2\, . \quad}
 \end{equation}
 \item A $\U(1)_a$ topological symmetry with flux $\ft_a$ contributes
 \begin{equation}
  {\quad \rule[-1.4em]{0pt}{3.4em} - \ft_a N^{3/2} \int dt\, \rho(t)\, t\, . \quad}
 \end{equation}
 \item A pair of bi-fundamental fields, one with magnetic flux $\fn_{(a,b)}$ and chemical potential $\Delta_{(a,b)}$ transforming
 in the $({\bf N},\overline{\bf N})$ of $\U(N)_a \times \U(N)_b$ and the other  with magnetic flux $\fn_{(b,a)}$ and chemical potential $\Delta_{(b,a)}$
 transforming in the $(\overline{\bf N},{\bf N})$ of $\U(N)_a \times \U(N)_b$, contributes
 \begin{equation}\label{bicontribution}
  {\quad \rule[-1.4em]{0pt}{3.4em} - N^{3/2} \int dt\, \rho(t)^2 \left[(\fn_{(b,a)}-1)\, g'_+ \left(\delta v(t) + \Delta_{(b,a)}\right)
  + (\fn_{(a,b)}-1)\, g'_- \left(\delta v(t) - \Delta_{(a,b)}\right)\right]\, . \quad}
 \end{equation}
 \item An adjoint field with magnetic flux $\fn_{(a,a)}$ and chemical potential $\Delta_{(a,a)}$, contributes
 \begin{equation}
  {\quad \rule[-1.4em]{0pt}{3.4em} - (\fn_{(a,a)}-1)\, g'_+ \left(\Delta_{(a,a)}\right) N^{3/2} \int dt\, \rho(t)^2\, . \quad}
 \end{equation}
 \item A field $X_a$ with magnetic flux $\fn_a$ transforming in the fundamental of $\U(N)_a$, contributes
 \begin{equation}
  {\quad \rule[-1.4em]{0pt}{3.4em} \frac{1}{2} (\fn_a - 1) N^{3/2} \int dt\, \rho(t) |t| \, , \quad}
 \end{equation}
 while an anti-fundamental field with magnetic flux $\tilde \fn_a$ contributes
 \begin{equation}
  {\quad \rule[-1.4em]{0pt}{3.4em} \frac{1}{2} (\tilde\fn_a - 1) N^{3/2} \int dt\, \rho(t) |t| \, . \quad}
 \end{equation}
 \item The tails,  where $\delta v$ has a constant value, as in \eqref{tails}, contribute
 \be\label{tailcontribution}
 - \fn_{(b,a)} N^{3/2} \int_{\delta v \approx - \Delta_{(b,a)} (\text{mod } 2\pi)}  dt \, \rho(t) Y_{(b,a)} - \fn_{(a,b)} N^{3/2} \int_{\delta v \approx  \Delta_{(a,b)} (\text{mod } 2\pi)}  dt \, \rho(t) Y_{(a,b)} \, ,
 \ee
 where the integral are taken on the tails regions.
    
\end{enumerate}

 As an example, for ABJM, using the above solution  of the BAEs, one obtains the simple expression \cite{Benini:2015eyy}
 \be  \re\log Z  = -\frac{N^{3/2} }{3} \sqrt{2 \Delta_1\Delta_2\Delta_3\Delta_4} \sum_a \frac{\fn_a}{\Delta_a}\, . \ee
  
 \section{Bethe potential versus free energy on $S^3$} \label{sec:freeenergy}

We would like to emphasize a remarkable connection of the large $N$ limit of the Bethe potential, which for us is an auxiliary quantity, with the large $N$ limit of the free energy $F$ on $S^3$ of the same ${\cal N}\geq 2$ theory.

Recall that the free energy $F$ on $S^3$ of an ${\cal N}=2$ theory is a function of trial R-charges $\Delta_I$  for the chiral fields \cite{Jafferis:2010un,Hama:2010av}. They parameterize the
curvature coupling of the supersymmetric Lagrangian on $S^3$. 
The $S^3$ free energy can  be computed using localization and reduced to a matrix 
model \cite{Kapustin:2009kz}. The large $N$ limit of the free energy, for $N \gg k_a$, has been computed in \cite{Herzog:2010hf,Jafferis:2011zi,Gulotta:2011vp,Gulotta:2011aa,Gulotta:2011si} and scales as $N^{3/2}$. For example, the free energy for
 ABJM with $k=1$ reads  \cite{Jafferis:2011zi}
 \be F = \frac{4 \pi N^{3/2}} {3} \sqrt{ 2 \Delta_1 \Delta_2 \Delta_3 \Delta_4} \, .
 \ee
We notice a striking similarity with \eqref{Vsolution sum 2pi -- end}. This is not a coincidence and generalizes to other theories. Indeed, remarkably, although the finite $N$ matrix models are quite different, for any ${\cal N}=2$ theory, the large $N$ limit of the Bethe potential becomes exactly equal to the large $N$ limit of the free energy $F$ on $S^3$. We can indeed compare the rules for constructing the Bethe potential with the rules for constructing the large $N$ limit of $F$, which have
been derived in  \cite{Jafferis:2011zi}. By comparing the rules in Section \ref{large N Bethe potential rules} with the rules given in Section 2.2 of  \cite{Jafferis:2011zi},
we observe that they are indeed the same up to a normalization. For reader's convenience the map is explicitly given in Table.\,\ref{Bethe map}.
\begin{table}[h!!!!]
\centering
\begin{tabular}{@{}l l@{}} \toprule \toprule
 Bethe potential & $S^3$ free energy\\
 \toprule
 $k_a$ & $- k_a$ \\
 $\mu$ & $\frac{\mu}{2}$ \\
 $v_a(t)$ & $\frac{v_a(t)}{2}$ \\
 $\rho(t)$ & $4 \rho(t)$ \\
 $\Delta_I$ & $\pi \Delta_I$ \\
 $\Delta_m$ & $- \pi \Delta_m$ \\
 $\mathcal{V}$ & $4 \pi i F$ \\
 $\mathcal{V} \Big|_{\text{BAEs}}$ & $\frac{i \pi}{2} F \Big|_{\text{On-shell}}$ \\[0.3cm]
 \toprule
 \toprule
\end{tabular}
\caption{The large $N$ Bethe potential versus the $S^3$ free energy of \cite{Jafferis:2011zi}.}
\label{Bethe map}
\end{table}
The conditions for cancellation of long-range forces (and therefore the allowed models) are also remarkably similar.

It might be surprising that our chemical potentials for global symmetries are mapped to R-charges in the free energy. However, remember that our $\Delta_I$ are angular variables. The invariance of the superpotential under the global symmetries implies that 
\be
\prod_{I \in \text{matter fields}} y_I = 1 \, ,
\ee
in each term of the superpotential, which 
is equivalent to 
\be \label{2pi}
\sum_{I \in \text{matter fields}} \Delta_I = 2 \pi \ell \,  \qquad \ell\in \mathbb{Z} \, ,
\ee 
where now $\Delta_I$ are the index chemical potentials.  Under the assumption $0<\Delta_I<2\pi$, few values of $\ell$ are actually allowed.
In the ABJM model reviewed above, only $\ell=1$ and $\ell=3$ give sensible results, with $\ell=3$ related to $\ell=1$ by a discrete symmetry of the model.
We found a similar result in  all the examples we have checked, and we do believe indeed that a solution of the BAE only exists when  
\be\label{superpotential0} 
\sum_{I \in \text{matter fields}} \Delta_I = 2 \pi\, ,
\ee
for each term of the superpotential, up to solutions related by discrete symmetries.
$\Delta_I / \pi$ then behaves at all effects like  a trial R-symmetry of the theory and we can  compare the index chemical potentials in $\cV$ with the R-charges in $F$.

\section{An index theorem for the twisted matrix model}\label{sec:index}

Under mild assumptions, the index at large $N$ can be actually extracted from the Bethe potential with a simple formula.

\begin{theorem}\label{entropy theorem}
The index of any $\cN \geq 2$ quiver gauge theory which respects the constraints \eqref{no long-range-forces Bethe0} and \eqref{no long-range forces}, and satisfies in addition \eqref{superpotential0}, can be written as
\begin{equation}\label{Z large N conjecture}
 {\quad \rule[-1.4em]{0pt}{3.4em}
 \re\log Z = - \frac{2}{\pi} \, \wb{\mathcal{V}}(\Delta_I) \,
 - \sum_{I}\, \left[ \left(\fn_I - \frac{\Delta_I}{\pi}\right) \frac{\partial \wb{\mathcal{V}}(\Delta_I)}{\partial \Delta_I}
  \right] \, ,
 \quad}
\end{equation}
where $\wb{\cV}$ is the extremal value of the functional \eqref{bethefunctional}
\begin{equation}\label{virial theorem}
 {\quad \rule[-1.4em]{0pt}{3.4em}
 \wb{\mathcal{V}}(\Delta_I) \equiv - i \mathcal{V} \Big|_{\text{BAEs}} = \frac{2}{3} \mu N^{3/2} \, ,}
\end{equation}
and $\mu$ is the  Lagrange multiplier appearing in \eqref{bethefunctional}.\footnote{The second identity in \eqref{virial theorem} is a consequence of a virial theorem for matrix models (see Appendix B of \cite{Gulotta:2011si}).}

\end{theorem}

\begin{proof}
We first replace the explicit factors of $\pi$, appearing in Eqs.\,\eqref{Bethe potential bi-fundamental}-\eqref{anti-fund cV rule}, with a formal variable $\bm{\pi}$.
Note that the ``on-shell'' Bethe potential $\wb{\mathcal{V}}$ is a homogeneous function of $\Delta_I$ and $\bm{\pi}$ and therefore it satisfies
\begin{equation}
 \wb{\mathcal{V}}(\lambda \Delta_I, \lambda \bm{\pi}) = \lambda^2 \, \wb{\mathcal{V}}(\Delta_I, \bm{\pi}) \quad \Rightarrow \quad
 \frac{\partial \wb{\mathcal{V}}(\Delta_I, \bm{\pi})}{\partial \bm{\pi}} =
 \frac{1}{\bm{\pi}} \left[ 2 \, \wb{\mathcal{V}}(\Delta_I) -\sum_I  \Delta_I \frac{\partial \wb{\mathcal{V}}(\Delta_I)}{\partial \Delta_I} \right]\, .
\end{equation}
Now, we consider a pair of bi-fundamental fields which contribute to the Bethe potential according to \eqref{Bethe potential bi-fundamental}.
The derivative of $\mathcal{V}(\Delta_I, \bm{\pi})$ with respect to  $\Delta_{(b,a)}$ and $\Delta_{(a,b)}$ is given by
 \begin{align}\label{proof bf}
  \sum_{I=(b,a),(a,b)} \fn_{I} \frac{\partial\mathcal{V}(\Delta_I, \bm{\pi})}{\partial\Delta_I} & =
  i N^{3/2} \int dt\, \rho(t)^2\, \left[ \fn_{(b,a)} g'_+ \left(\delta v(t) + \Delta_{(b,a)}\right) + \fn_{(a,b)} g'_- \left(\delta v(t) - \Delta_{(a,b)}\right)\right] \nn \\
  & +   \sum_{I=(b,a),(a,b)}  \fn_{I}  \underbrace{\frac{\partial \mathcal{V}}{\partial \rho} \frac{\partial \rho}{\partial \Delta_I}}_\text{vanishing on-shell}
  +   \sum_{I=(b,a),(a,b)} \fn_{I}  \underbrace{\frac{\partial \mathcal{V}}{\partial (\delta v)} \frac{\partial (\delta v)}{\partial \Delta_I}}_\text{tails contribution} \, .
 \end{align}
 The expression in the first line is precisely part of the contribution of a pair of bi-fundamentals \eqref{bicontribution} to the index.  
In the tails, using \eqref{tails}, we find
 \begin{equation}
 \frac{\partial (\delta v)}{\partial \Delta_{(b,a)}} = -1 \, , \qquad
 \frac{\partial (\delta v)}{\partial \Delta_{(a,b)}} = 1 \, , \qquad
 \frac{\partial \mathcal{V}}{\partial (\delta v)} = - i Y_{(b,a)} \rho \, , \qquad
 \frac{\partial \mathcal{V}}{\partial (\delta v)} = i Y_{(a,b)} \rho \, . 
 \end{equation}
Therefore, the last term in Eq.\,\eqref{proof bf} can be simplified to
 \begin{equation}
 i N^{3/2}  \fn_{(b,a)} \int_{\delta v \approx - \Delta_{(b,a)}} dt\, \rho(t)\, Y_{(b,a)}
 + i N^{3/2}  \fn_{(a,b)}  \int_{\delta v \approx \Delta_{(a,b)}} dt\, \rho(t)\, Y_{(a,b)} \, .
\end{equation}
This precisely gives the tail contribution \eqref{tailcontribution} to the index. 
Next, we take the derivative of the Bethe potential with respect to  $\bm{\pi}$. It can be written as
 \begin{align}\label{pi derivative}
 \frac{\partial \mathcal{V}}{\partial \bm{\pi}}  & = - i N^{3/2} \int dt\, \rho(t)^2 \left[ g'_+ \left(\delta v(t) + \Delta_{(b,a)}\right) + g'_- \left(\delta v(t) - \Delta_{(a,b)}\right) \right] \nn \\
 & + i N^{3/2} \int dt\, \rho(t)^2\, \left[ \frac{2 \bm{\pi}^2}{3} - \frac{\bm{\pi}}{3} \left( \Delta_{(b,a)} + \Delta_{(a,b)} \right) \right] \, .
 \end{align}
 The expression in the first line completes the contribution of a pair of bi-fundamentals \eqref{bicontribution} to the index.  
The expression  in the second line, after summing  over all the bi-fundamental pairs,  can be written as
 \begin{equation}
\sum_{\text{pairs}} \left[ \frac{2 \bm{\pi}^2}{3} - \frac{\bm{\pi}}{3} \left( \Delta_{(b,a)} + \Delta_{(a,b)} \right)\right]
 = \frac{\bm{\pi}}{3} \sum_{I} \left( \bm{\pi} - \Delta_I \right)
 = \frac{\bm{\pi}^2}{3} |G| \, , 
 \end{equation}
which is precisely the contribution of the gauge fields  \eqref{gaugecontribution} to the index.  
Here, we used the condition
 \begin{equation}\label{pi constraint}
 \bm{\pi} |G| + \sum_{I} \left( \Delta_I - \bm{\pi} \right) = 0 \, .
 \end{equation}
This condition follows from the fact that, assuming  \eqref{superpotential0} for each superpotential term, $\Delta_I/\pi$ behaves as a trial R-symmetry, so that  \eqref{no long-range forces} yields
\begin{equation}\label{no long-range forces2}
 2 + \sum_{I \in a} \left(\frac{\Delta_I}{\pi} - 1\right) = 0 \, ,
\end{equation}
which, summed over all the nodes, since each bi-fundamental field belongs precisely to two nodes, gives \eqref{pi constraint}.
Condition \eqref{pi constraint} is indeed equivalent to ${\rm Tr} R=0$, where the trace is taken over the bi-fundamental  fermions and gauginos in the quiver and $R$ is an R-symmetry. 
Combining everything as in the right hand side of Eq.\,\eqref{Z large N conjecture} we obtain the contribution of gauge and bi-fundamental fields to the index.
The proof for all the other matter fields and the topological symmetry is straightforward.
\end{proof}

If we ignore the linear relation among the chemical potentials, we can always use a  set of  $\Delta_I$ such that $\wb{\cV}$ is a homogeneous function of degree two of the
$\Delta_I$ alone.\footnote{This is what happens in \eqref{Vsolution sum 2pi -- end} for ABJM. Recall that $\sum_i \Delta_i = 2 \pi$ so that the four $\Delta_i$ are not linearly independent.}
In this case, the index theorem simplifies to
\begin{equation}\label{Z large N conjecture2}
 {\quad \rule[-1.4em]{0pt}{3.4em}
 \re\log Z = 
 - \sum_{I}\, \fn_I \frac{\partial \wb{\mathcal{V}}(\Delta_I)}{\partial \Delta_I}
  \, .
 \quad}
\end{equation}

\section{Theories with $N^{5/3}$ scaling of the index}\label{N53}

Chern-Simons quivers of the form \eqref{quiver} have a rich parameter space. If  condition \eqref{CScontr} is satisfied and $N\gg k_a$,  they have an M-theory weakly coupled dual.  In the t'Hooft limit, $N,k_a\gg 1$ with $N/k_a =\lambda_a$ fixed and large, they have a type IIA weakly coupled dual. When instead
\be\label{CScontr2}
\sum_{a=1}^{|G|} k_a \ne 0 \, ,
\ee
they probe massive type IIA \cite{Gaiotto:2009mv}. There is an interesting limit, given \eqref{CScontr2}, where again   $N\gg k_a$. The limit is no more an M-theory phase \cite{Aharony:2010af},
but rather an extreme phase of massive type IIA. Supergravity duals of this type of phases have been found in \cite{Petrini:2009ur,Lust:2009mb,Aharony:2010af,Tomasiello:2010zz,Guarino:2015jca,Fluder:2015eoa,Pang:2015vna,Pang:2015rwd}. The free energy scales as $N^{5/3}$ \cite{Aharony:2010af}. We now show that also the topologically twisted index scales in the 
same way.  As it happens for the $S^3$ matrix model \cite{Jafferis:2011zi,Fluder:2015eoa}, we find a consistent large $N$ limit whenever the constraints \eqref{globalJ} and \eqref{globalR} are satisfied.

The ansatz for the eigenvalue distribution is now,  as in \cite{Jafferis:2011zi,Fluder:2015eoa},
 \begin{equation}\label{ansatz N53}
 u^{(a)} (t) = N^{1/3} (i t + v(t))\, .
\end{equation}
The scaling is again dictated by the competition between the Chern-Simons terms, now with \eqref{CScontr2},  and the gauge and bi-fundamental contributions.

\subsection{Long-range forces}
Since the eigenvalue distribution is the same for all gauge groups, the long-range forces \eqref{N3} cancel automatically. We see that, differently from before, we can have a consistent large $N$ limit also in the case of  chiral quivers. We also need  to cancel the long-range forces  \eqref{N2ang}.  They compensate each other  if condition  \eqref{no long-range-forces Bethe0}  is satisfied.
Since the eigenvalues are the same for all groups, it is actually  enough to sum over nodes and we obtain the milder constraint  \eqref{globalJ} on the flavor charges:
\begin{equation} \Tr J =0\, ,
\end{equation}
where the trace is taken over bi-fundamental fermions in the quiver.

We obtain similar conditions by looking at the scaling of the twisted index.
As  in Section \ref{proof vanishing forces},  vector multiplets and chiral bi-fundamental multiplets contribute terms \eqref{indexscalingg} and \eqref{indexscalingb} which are of order $\cO(N^{7/3})$. They compensate each other  if condition \eqref{no long-range forces} is satisfied.
Since the eigenvalues are the same for all groups, it is again   enough to sum over nodes and we obtain the  constraint \eqref{globalR} on the flavor magnetic fluxes:
\begin{equation}
\Tr R =|G| + \sum_{I} \left(\fn_I - 1\right) = 0\, ,
\end{equation}
where the trace is taken over bi-fundamental fermions and gauginos in the quiver. 

Conditions $\Tr R=\Tr J=0$ are certainly satisfied for all quivers with a four-dimensional parent, even the chiral ones.

\subsection{Bethe potential at large $N$}

\begin{enumerate}
 \item Each group $a$ with CS level $k_a$ contributes
 \begin{equation}
  {\quad \rule[-1.4em]{0pt}{3.4em} - i k_a N^{5/3} \int dt\, \rho(t)\, t\, v(t) + \frac{k_a}{2} N^{5/3} \int dt\, \rho(t)\, \left(t^2 - v(t)^2\right)\, . \quad}
 \end{equation}
 \item A  bi-fundamental field with chemical potential $\Delta_I$  contributes
 \begin{equation}
  {\quad \rule[-1.4em]{0pt}{3.4em} i  \, g_+\left(\Delta_{I}\right) \, N^{5/3} \int dt\, \frac{\rho(t)^2}{1-i v'(t)} \, . \quad}
 \end{equation}
 \item A fundamental field contributes
 \begin{equation}
  {\quad \rule[-1.4em]{0pt}{3.4em} -\frac{1}{4} N^{5/3} \int dt\, \rho(t)\, \sign(t) \left[ i t + v(t) \right]^2\, , \quad}
 \end{equation}
 while an anti-fundamental field contributes
 \begin{equation}
  {\quad \rule[-1.4em]{0pt}{3.4em} \frac{1}{4} N^{5/3} \int dt\, \rho(t)\, \sign(t) \left[ i t + v(t) \right]^2 \, . \quad}
 \end{equation}
\end{enumerate}

\subsection{The index at large $N$}
\begin{enumerate}
 \item For each group $a$, the contribution of the Vandermonde determinant is
 \begin{equation}
  {\quad \rule[-1.4em]{0pt}{3.4em} -\frac{\pi^2}{3} N^{5/3} \int dt\, \frac{\rho(t)^2}{1-i v'(t)}\, . \quad}
 \end{equation}
 \item  A chiral  bi-fundamental field, with chemical potential $\Delta_I$ and magnetic flux $\fn_I$ contributes
 \begin{equation}
  {\quad \rule[-1.4em]{0pt}{3.4em} - (\fn_{I}-1)\, g'_+ \left(\Delta_{I}\right) N^{5/3} \int dt\, \frac{\rho(t)^2}{1-i v'(t)} \, . \quad}
 \end{equation}
 Fundamental fields do not contribute to the index explicitly.
\end{enumerate}
Notice that the relation with the $S^3$ free energy discussed in Section \ref{sec:freeenergy} and the {\em index theorem} of Section \ref{sec:index} also hold for this class of quiver gauge theories.\footnote{The coefficient $2/3$ in front of $\mu$ in Eq.\,\eqref{virial theorem} must be replaced by $3/5$.}

\section{Discussion and Conclusions}\label{discussion}

In this paper we have studied the large $N$ behavior of the topologically twisted index for ${\cal N}=2$ gauge theories in three dimensions. We have focused on theories  with a conjectured M-theory or massive type IIA dual and examined the corresponding field theory phases, where holography predicts a $N^{3/2}$ or $N^{5/3}$ scaling for the path integral, respectively. We correctly reproduced this scaling  for a class of ${\cal N}=2$ theories and we also uncovered some surprising relations with apparently
different physical quantities.

The first surprise comes from the identification of the {\it Bethe potential} $\wb\cV$ with the {\it $S^3$ free energy} $F$ of the same  ${\cal N}=2$ gauge theory. Recall that, in our
approach, the BAE and the Bethe potential are auxiliary quantities determining the position of the poles in the matrix model in the large $N$ limit.  $\wb\cV$ depends on the chemical potentials for the flavor symmetries, satisfying \eqref{2pi}, while $F$ depends on trial R-charges, which parameterize the curvature couplings of the theory on $S^3$.
Both quantities, $\cV$ and $F$ are determined in terms of a matrix model (auxiliary in the case of $\cV$). The two matrix models, and the corresponding equations of motion
are different for finite $N$, but, quite remarkably become indistinguishable in the large $N$ limit. Also the conditions to be imposed on the quiver for the existence of a  $N^{3/2}$ or $N^{5/3}$ scaling are the same. Although the structure of the long-range forces and the mechanism for their cancellation is different, they rule out
quivers with chiral bi-fundamentals in the M-theory phase and impose the same conditions on flavor symmetries. 

This identification leads to a relation of the Bethe potential $\wb\cV$ with the volume functional of Sasaki-Einstein manifolds.
The exact R-symmetry of a superconformal  ${\cal N}=2$ gauge theory can be found by extremizing $F(\Delta_I)$ with respect to the trial R-charges $\Delta_I$ \cite{Jafferis:2011zi}, but $F(\Delta_I)$ makes sense for arbitrary $\Delta_I$. The functional $F(\Delta_I)$ has a well-defined geometrical meaning for theories with an AdS$_4\times Y_7$ dual, where $Y_7$ is a Sasaki-Einstein manifold. The  value of $F$  upon extremization is related to the (square root of the) volume of $Y_7$. More generally, at least for a class of quivers corresponding to ${\cal N}=3$ and toric cones $C(Y_7)$, the  value of $F(\Delta_I)$, as a function of the trial R-symmetry parameterized by $\Delta_I$,  has been matched with the (square root of the) volume of a family of Sasakian deformation of $Y_7$, as a function of the Reeb vector. For toric theories, the volume can be parameterized in terms of a set of charges $\Delta_I$, that encode how the R-symmetry  varies  with the Reeb vector, and it has been conjectured in \cite{Amariti:2011uw,Amariti:2012tj,Lee:2014rca}
to be a homogeneous quartic function of the $\Delta_I$, in agreement with the homogeneity properties of $\wb\cV$ and $F$. 
  
A second intriguing relation comes from the index theorem \eqref{Z large N conjecture}. The original reason for studying the large $N$ limit of the topologically twisted index comes from
the counting of AdS$_4$ black holes microstates. The entropy of magnetically charged black holes asymptotic to AdS$_4\times S^7$ was successfully compared with
the large $N$ limit of the index in  \cite{Benini:2015eyy}, when extremized with respect to the chemical potential $\Delta_I$. We expect that a similar relation holds for magnetically charged BPS black holes asymptotic to AdS$_4\times Y_7$,
for a generic Sasaki-Einstein manifold. Given the very small number of black holes known, this statement is difficult to check. Assuming, however, that it is true, we can 
compare the on-shell index of  a superconformal  ${\cal N}=2$ gauge theory dual to AdS$_4\times Y_7$ twisted by a set of magnetic charges $\fn_I$ with the entropy of a black hole in AdS$_4\times Y_7$ supported by the same magnetic charges. The entropy of such black hole is determined in supergravity by the attractor mechanism \cite{Ferrara:1996dd}.  The black
hole can be written as a solution of the ${\cal N}=2$ gauged supergravity obtained by truncating the KK spectrum on $Y_7$ to a consistent set of modes, which contains
vector and hypermultiplets \cite{Cacciatori:2009iz,Dall'Agata:2010gj,Hristov:2010ri,Katmadas:2014faa,Halmagyi:2014qza}. In a gauged supergravity with only vectors, the entropy of the black hole can be obtained  by extremizing with respect to $X^I$ the quantity \cite{Dall'Agata:2010gj}
\begin{equation}\label{attractor}
\cI(X^I) = - \sum_I  \fn^I \frac{\partial  \cal F}{\partial X^I} \, ,
\end{equation}
where ${\cal F}(X^I)$ is the supergravity prepotential and $X^I$ a set of covariantly-constant homogeneous holomorphic sections. Here, we are working in the gauge $\sum_I g_I X^I=1$, where $g_I$ are the electric gaugings of the theory and we assume that there are no magnetic ones.  The presence of hypermultiplets
just add algebraic constraints \cite{Halmagyi:2013sla,Klemm:2016wng}. Comparison of the attractor equation \eqref{attractor} with the index theorem \eqref{Z large N conjecture2}
suggests the identification of $\Delta_I$ with $X^I$ and a proportionality between  $\wb\cV(\Delta_I)$ and ${\cal F}(X^I)$, valid also before extremization. This proportionality
certainly holds for ABJM since the prepotential is
\begin{equation}
{\cal F} = i \sqrt{ X^0 X^1 X^2 X^3}\, ,
\end{equation} 
which can be clearly mapped to  \eqref{Vsolution sum 2pi -- end}. It would be quite interesting to see how to formulate this identification in more general theories with
hypermultiplets.

We thus see an intriguing chain of identifications
\\

{\it Bethe potential} $\wb\cV \qquad  \Longleftrightarrow \qquad S^3$ {\it free-energy} $F \qquad \Longleftrightarrow \qquad $ {\it prepotential} ${\cal F}$
\\

\noindent of functionals depending on chemical potentials, trial R-charges and bulk scalar fields, respectively. 
A relation between the free energy $F$ and the prepotential of the compactified theory was already suggested in \cite{Lee:2014rca}. 
This chain of identifications certainly calls for further investigation.\footnote{As well as the relation to other extremization problems  and generalizations \cite{Amariti:2015ybz,Hristov:2016vbm}.}

 The main motivation of our analysis comes certainly from the attempt to extend the result of  \cite{Benini:2015eyy} to a larger class of black holes. The difficulty of doing so
 is mainly the exiguous number of  existing black holes solutions with an M-theory lift. Few numerical examples are known in Sasaki-Einstein compactifications \cite{Halmagyi:2013sla}, 
 mostly having Betti multiplets as massless vectors. Some interesting examples involves chiral quivers and are therefore outside the range of our 
 technical abilities at the moment. It is curious that  apparently well-defined chiral quivers, which passed quite nontrivial checks \cite{Benini:2011cma}, have an ill-defined 
large $N$ limit both for the $S^3$ free energy and the topologically twisted index in the M-theory phase. It would be quite interesting to know whether  this is just a technical problem and another saddle-point with $N^{3/2}$ scaling exists, or the models are really ruled out.  

It would be also  quite interesting to find new examples of AdS$_4$ M-theory and massive type IIA  black holes directly in eleven or ten dimensions (see, for example, \cite{Katmadas:2015ima}) or in some other consistent truncations of eleven dimensional supergravity where to test our results. 

We hope to come back to all these questions quite soon.

\section*{Acknowledgements}

We would like to thank Alessandra Gnecchi,  Anton Nedelin  and especially Fancesco Benini and Kiril Hristov for useful discussions. We thank Noppadol Mekareeya for many useful discussions and collaboration on a related project. AZ is  grateful to CERN for hospitality and its partial support during the completion of this work. 
AZ is supported by the INFN and the MIUR-FIRB grant RBFR10QS5J ``String Theory and Fundamental Interactions''.
SMH is supported in part by INFN.

\appendix

\section{Derivation of general rules for theories with $N^{3/2}$ scaling of the index}
\label{general rules N32}

In this appendix we give a detail derivation of the rules presented in the main text for finding the Bethe potential and the index at large $N$.
The result is a straightforward generalization of \cite{Benini:2015eyy}.

We consider the following large $N$ saddle-point eigenvalue distribution ansatz
\be\label{ansatz alpha}
u_i^{(a)} = i N^{\alpha} t_i + v_i^{(a)} \, .
\ee
Notice that the imaginary parts of the $u_i^{(a)}$ are equal. We also define
\be
\delta v_i = v_i^{(b)} - v_i^{(a)} \, .
\ee
In the large $N$ limit, we define the continuous functions $t_i = t(i/N)$ and $v_i^{(a)}= v^{(a)}(i/N)$ and  we introduce the density of eigenvalues 
\be \rho(t) = \frac1N\, \frac{di}{dt}\, , \ee
normalized as $ \int dt\, \rho(t) = 1$. 
Furthermore, we impose the additional constraint
\be
\sum_{a=1}^{|G|} k_a = 0 \, ,
\ee
corresponding to quivers dual to M-theory on AdS$_4\times Y_7$ and $N^{3/2}$ scaling. 

\subsection{Bethe potential at large $N$}
\label{Bethe potential rules N32}

We may write the Bethe ansatz equations as
\begin{align}\label{log BAEs2}
 0 & = \log \left[ \text{RHS of \eqref{BA expression}} \right] - 2 \pi i n_i^{(a)} \, ,
\end{align}
where $n_i^{(a)}$ are integers that parameterize the angular ambiguities.
We define the ``Bethe potential'' as the function whose critical points gives the BAEs \eqref{log BAEs2}.
In  the large $N$ limit the Bethe potential $\cV$ will be the sum of various contributions,
\be
\cV = \cV^{\text{CS}} + \cV^{\text{bi-fund}} + \cV^{\text{adjoint}} + \cV^{\text{(anti-)fund}} \, .
\ee
$\alpha$ will be determined to be $1/2$ by the competition between Chern-Simons terms and matter contribution.

\subsubsection{Chern-Simons contribution}
\label{Bethe CS N32}

Each group $a$ with CS level $k_a$ and topological chemical potential $\Delta_m^{(a)}$, contributes to the  finite $N$ Bethe potential as
\begin{equation}
 \cV^{\text{CS}} = \sum_{i=1}^N  \left[- \frac {k_a}{2} \left(u_i^{(a)}\right)^2 - \Delta_m^{(a)} u_i^{(a)} \right] \, .
\end{equation}
Given the large $N$ saddle-point eigenvalue distribution \eqref{ansatz alpha}, we find
\begin{equation}
 \cV^{\text{CS}} = \frac{k_a}{2} N^{2\alpha} \sum_{i=1}^{N} t_i^2 - i N^{\alpha} \sum_{i=1}^{N} \left(k_a t_i v_i^{(a)} + \Delta_m^{(a)} t_i \right) \, .
\end{equation}
Summing over nodes the first term vanishes (since $\sum_{a=1}^{|G|} k_a = 0$).
Taking the continuum limit, we obtain
\begin{equation}
 \cV^{\text{CS}} = - i k_a N^{1+\alpha} \int dt\, \rho(t)\, t\, v_a(t) - i N^{1+\alpha} \Delta_m^{(a)} \int dt\, \rho(t)\, t\, .
\end{equation}

\subsubsection{Bi-fundamental contribution}
\label{Bethe bi-fundamentals N32}

For a pair of bi-fundamental fields, one with chemical potential $\Delta_{(a,b)}$ transforming in the $({\bf N},\overline{\bf N})$
of $\U(N)_a \times \U(N)_b$ and one with chemical potential $\Delta_{(b,a)}$ transforming in the $(\overline{\bf N},{\bf N})$ of $\U(N)_a \times \U(N)_b$, the finite $N$  contribution to the Bethe potential is given by
\begin{align}\label{cV bi-fundamentals}
 \cV^{\text{bi-fund}} & = \sum_{\substack{\text{bi-fundamentals} \\ (b,a) \text{ and } (a,b)}} \sum_{i,j=1}^N \left[ \Li_2 \left( e^{i \left(u_j^{(b)} - u_i^{(a)} + \Delta_{(b,a)}\right)} \right) - \Li_2 \left( e^{i \left(u_j^{(b)} - u_i^{(a)} - \Delta_{(a,b)}\right)} \right) \right]  \nn \\
 & - \sum_{\substack{\text{bi-fundamentals} \\ (b,a) \text{ and } (a,b)}} \sum_{i,j=1}^N  \left[ \frac{\left( \Delta_{(b,a)} - \pi \right) + \left( \Delta_{(a,b)} - \pi \right) }{2} \left( u_j^{(b)} - u_i^{(a)} \right) \right] \, ,
\end{align}
up to constants that do not depend on $u_j^{(b)}$, $u_i^{(a)}$.

We would like to remind the reader that all angular variables are defined modulo $2\pi$. Part of the ambiguity in $\Delta_I$ can be fixed by requiring that
\be
\label{inequalities for delta v}
0 < \delta v + \Delta_{(b,a)} < 2\pi \;,\qquad\qquad\qquad -2\pi < \delta v - \Delta_{(a,b)} < 0 \;.
\ee
The remaining ambiguity of simultaneous shifts $\delta v \to \delta v+2\pi$, $\Delta_{(a,b)} \to \Delta_{(a,b)} +2\pi$, $\Delta_{(b,a)} \to \Delta_{(b,a)}-2\pi$ can also be fixed by requiring that $\delta v(t)$ takes the value $0$ somewhere, if it vanishes at all, which we assume.
We then have
\be
\label{general range of Delta_a}
0 < \Delta_I < 2\pi \, .
\ee

To compute $\cV^{\text{bi-fund}}$, we break
\begin{align}
\label{broken expression}
\sum_{i,j=1}^N \Li_2\left( e^{i \left(u_j^{(b)} - u_i^{(a)}+ \Delta_{(b,a)}\right)} \right) & = \sum_{i>j} \Li_2\left( e^{i \left(u_j^{(b)} - u_i^{(a)} + \Delta_{(b,a)}\right)} \right) + \sum_{i<j} \Li_2\left( e^{i \left( u_j^{(b)} - u_i^{(a)} + \Delta_{(b,a)}\right)} \right) \nn \\
& + \sum_{i=1}^N \Li_2\left( e^{i \left(u_i^{(b)} - u_i^{(a)} + \Delta_{(b,a)}\right)} \right) \, .
\end{align}
The crucial point here is that the last term is naively of $\cO(N)$ and thus subleading; however, we should keep it since its derivative is not subleading on part of the solution when $\delta v$ hits $\Delta_{(a,b)}$ or $-\Delta_{(b,a)}$.
Therefore, we keep
\be\label{talis eom}
N \int dt\, \rho(t) \left[\Li_2 \left( e^{i \left( \smallstrut \delta v(t) + \Delta_{(b,a)} \right)} \right) - \Li_2 \left( e^{i \left( \smallstrut \delta v(t) - \Delta_{(a,b)} \right)} \right) \right] \; .
\ee
This will be important in  the {\em tail contribution} to the Bethe potential.
The second term in \eqref{broken expression} is
\be
\sum_{i<j} \Li_2\left( e^{i \left(u_j^{(b)} - u_i^{(a)} + \Delta_{(b,a)} \right)} \right) = N^2 \int dt\, \rho(t) \int_t dt'\, \rho(t')\, \Li_2 \left( e^{i \left( \smallstrut u_b(t') - u_a(t) + \Delta_{(b,a)} \right)} \right) \;.
\ee
We first write the dilogarithm function as a power series, \ie,
\be
\Li_2(e^{iu}) = \sum_{k=1}^\infty \frac{e^{iku}}{k^2} \, .
\ee
Then, we consider the integral
\be
I_k = \int_t dt'\, \rho(t')\, e^{ik \left( \smallstrut u_b(t') - u_a(t) + \Delta_{(b,a)} \right)} = \int_t dt'\, e^{-kN^\alpha (t'-t)} \sum_{j=0}^\infty \frac{(t'-t)^j}{j!} \partial_x^j \left[ \rho(x)\, e^{ik \left( \smallstrut v_b(x) - v_a(t) + \Delta_{(b,a)} \right)} \right]_{x=t} \;, \nonumber
\ee
where in the second equality we have Taylor-expanded the integrand around the lower bound.
Doing the integration over $t'$ we see that the leading contribution is for $j=0$, thus
\be
I_k = \frac{\rho(t)\, e^{ik \left( \smallstrut v_b(t) - v_a(t) + \Delta_{(b,a)} \right)} }{ k N^\alpha} + \cO(N^{-2\alpha}) \, .
\ee
Substituting we find
\be
\label{first summation}
\sum_{i<j} \Li_2 \left( e^{i \left( u_j^{(b)} - u_i^{(a)} + \Delta_{(b,a)} \right)} \right) = N^{2-\alpha} \int dt\, \Li_3 \left( e^{i \left( \smallstrut \delta v(t) + \Delta_{(b,a)} \right)} \right)\, \rho(t)^2 + \cO(N^{2-2\alpha}) \;.
\ee
Next, we need to compute the first term in \eqref{broken expression}.
In order for the integral to be localized at the boundary, we need to invert the integrand.
Since $0<\re \left(u_j^{(b)} - u_i^{(a)} + \Delta_{(b,a)}\right)<2\pi$: 
\begin{align}
\label{reflection applied 34}
\Li_2 \left( e^{i \left( u_j^{(b)} - u_i^{(a)} + \Delta_{(b,a)} \right)} \right) & = - \Li_2 \left( e^{i \left( u_i^{(a)} - u_j^{(b)} - \Delta_{(b,a)} \right)} \right) + \frac{ \left(u_j^{(b)} - u_i^{(a)} + \Delta_{(b,a)}\right)^2}2 \nn \\
& - \pi \left(u_j^{(b)} - u_i^{(a)} + \Delta_{(b,a)}\right) + \frac{\pi^2}3 \;.
\end{align}
The summation $\sum_{i>j}$ of the first term in the latter expression is similar to \eqref{first summation} but with $-\Li_3\left( e^{-i \left( \delta v(t) + \Delta_{(b,a)} \right)} \right)$ instead of $\Li_3$.
The two contributions may then be combined, using \eqref{reflection formulae},
\begin{align}
& N^{2-\alpha} \int dt\,  \left[\Li_3 \left( e^{i \left( \smallstrut \delta v(t) + \Delta_{(b,a)} \right)} \right) - \Li_3\left( e^{-i \left( \delta v(t) + \Delta_{(b,a)} \right)} \right) \right]\, \rho(t)^2 \nn \\
& = i N^{2-\alpha} \int dt\,  g_+\left( \smallstrut \delta v(t) + \Delta_{(b,a)} \right) \, \rho(t)^2\, ,
\end{align}
where we have introduced the polynomial function $g_+(u)$ defined in  Eq.\,\eqref{gp gm}.

The second term in the first line of (\ref{cV bi-fundamentals}) can be treated similarly.
We now have $-2\pi< \re (u_j^{(b)} - u_i^{(a)} - \Delta_{(a,b)}) < 0$ and
\begin{align}
\label{reflection applied 12}
- \Li_2 \left( e^{i \left(u_j^{(b)} - u_i^{(a)} - \Delta_{(a,b)}\right)} \right) & = \Li_2 \left( e^{i \left(u_i^{(a)} - u_j^{(b)} + \Delta_{(a,b)} \right)} \right) - \frac{\left(u_j^{(b)} - u_i^{(a)} - \Delta_{(a,b)}\right)^2}2 \nn \\
& - \pi \left(u_j^{(b)} - u_i^{(a)} - \Delta_{(a,b)} \right) - \frac{\pi^2}3 \, .
\end{align}
As before, the result of the summation $\sum_{i>j}$ together with that of $\sum_{i<j}$ yields a cubic polynomial expression
\begin{align}
&  - i N^{2-\alpha} \int dt\,  g_-\left( \smallstrut \delta v(t) - \Delta_{(a,b)} \right) \, \rho(t)^2\, ,
\end{align}
where $g_-(u)$ is defined in  Eq.\,\eqref{gp gm}.

The left over terms from \eqref{reflection applied 34} and \eqref{reflection applied 12}, throwing away the constants which do not affect the critical points, are
\be\label{bi-fundamentals forces}
 \left[ \left( \Delta_{(a,b)} - \pi \right) + \left( \Delta_{(b,a)} - \pi \right) \right] \sum_{i>j} \left( u_j^{(b)} - u_i^{(a)}\right) \, ,
\ee
which, combined with the second line in \eqref{cV bi-fundamentals}, gives
\be\label{bi-fundamentals forces2}
 \frac{1}{2} \left[ \left( \Delta_{(a,b)} - \pi \right) + \left( \Delta_{(b,a)} - \pi \right) \right] \sum_{i\neq j} \left( u_j^{(b)} - u_i^{(a)}\right) \text{sign} (i-j)\, .
\ee
This term can be precisely canceled by
\begin{equation}
 - 2 \pi \sum_{i=1}^N \left( n_i^{(b)} u_i^{(b)} - n_i^{(a)} u_i^{(a)} \right) \, ,
\end{equation}
provided that $\sum_{I \in a} \Delta_I \in 2 \pi \bZ $.\footnote{When $N \in 2 \bZ_{\geq 0} +1 $. For even $N$  one can include an extra $(-1)^{\fm}$ in the twisted partition function, which can be reabsorbed in the definition of the topological fugacity $\xi$, to compensate the overall factor of $\pi$.}

Notice that  a single bi-fundamental chiral multiplet, with chemical potential $\Delta_{(b,a)}$,
transforming in the representation $(\overline{\bf N},{\bf N})$ of $\U(N)_a \times \U(N)_b$ contributes to the Bethe potential as
\begin{equation}\label{single bifund Bethe}
 \sum_{i,j=1}^N \left[ \Li_2 \left( e^{i \left( u_j^{(b)} - u_i^{(a)} + \Delta_{(b,a)} \right)} \right) - \frac{\left( u_j^{(b)} - u_i^{(a)} + \Delta_{(b,a)} \right)^2}{4} \right] \, .
\end{equation}
Using Eq.\,\eqref{reflection applied 34} we find the following long-range terms
\be\label{forces chiral}
\sum_{i<j} \frac{\left(u_i^{(a)} - u_j^{(b)}\right)^2}{4} - \sum_{i<j} \frac{\left(u_j^{(a)} - u_i^{(b)}\right)^2}{4}\, .
\ee
In the large N limit, they are of order $N^{5/2}$ and cannot be canceled  for {\it chiral quivers}.

To find a nontrivial saddle-point the leading terms of order $N^{1+\alpha}$ and $N^{2-\alpha}$ have to be of the same order, so we need $\alpha = 1/2$.
Putting everything together we arrive at the final expression for the large $N$ contribution of the bi-fundamental fields to the Bethe potential
\begin{equation}
 \cV^{\text{bi-fund}} = i N^{3/2} \sum_{\substack{\text{bi-fundamentals} \\ (b,a) \text{ and } (a,b)}} \int dt\, \rho(t)^2 \left[g_+\left(\delta v(t) + \Delta_{(b,a)}\right) - g_-\left(\delta v(t) - \Delta_{(a,b)}\right)\right]\, .
\end{equation}

In the sum over pairs of bi-fundamental fields $(b,a)$ and $(a,b)$, adjoint fields should be counted once and should come with an explicit factor of $1/2$.
Keeping this in mind and setting
\begin{equation}
v_b = v_a \, , \qquad \qquad  \Delta_{(b,a)} =  \Delta_{(a,b)} =  \Delta_{(a,a)} \, ,
\end{equation}
we find the contribution of fields transforming in the adjoint of the $a$th gauge group with chemical potential $\Delta_{(a,a)}$ to the large $N$ Bethe potential,
\begin{equation}
 \cV^{\text{adjoint}} = i N^{3/2} \sum_{\substack{\text{adjoint} \\ (a,a) }} g_+\left(\Delta_{(a,a)}\right) \int dt\, \rho(t)^2 \, .
\end{equation}

\subsubsection{Fundamental and anti-fundamental contribution}
\label{Bethe fundamentals N32}

The fundamental and anti-fundamental fields contribute to the large $N$  Bethe potential as\footnote{Up to a factor $-\pi (\tilde n_a -n_a) u_i/2$ that cancels at this order  for 
total number of fundamentals equal to total number of anti-fundamentals, which we will need to assume for consistency.} 
\begin{align}
 \cV^{\text{(anti-)fund}} & =
 \sum_{i=1}^{N} \left[ \sum_{\substack{\text{anti-fundamental} \\ a }} \Li_2 \left(e^{i\left(-u_i ^{(a)} + \tilde \Delta_a \right)}\right) - \sum_{\substack{\text{fundamental} \\ a }} \Li_2 \left(e^{i\left(-u_i ^{(a)} - \Delta_a\right)}\right) \right]\nn\\
 & + \frac{1}{2} \sum_{i=1}^{N} \left[ \sum_{\substack{\text{anti-fundamental} \\ a }} \big( \tilde \Delta_a - \pi \big) u_i^{(a)} + \sum_{\substack{\text{fundamental} \\ a }} \big( \Delta_a - \pi \big) u_i^{(a)}\right] \nn \\
 & - \frac{1}{4} \sum_{i=1}^{N} \left[\sum_{\substack{\text{anti-fundamental} \\ a }} \left(u_i^{(a)}\right)^2
 - \sum_{\substack{\text{fundamental} \\ a }} \left(u_i^{(a)}\right)^2 \right] \, .
\end{align}
Let us denote the total number of (anti-)fundamental fields by $(\tilde n_a)\, n_a$.
Substituting in $\cV^{\text{(anti-)fund}}$ the ansatz \eqref{ansatz alpha} and taking the continuum limit, the first line contributes
\begin{align}\label{fund Li2}
 & - \frac{\left( \tilde n_a - n_a \right)}{2} N^{2} \int_{t>0} dt\, \rho(t)\, t^2 \nn \\
 & + i N^{3/2} \int_{t>0} dt\, \rho(t)\, t\, \left\{ \sum_{\substack{\text{anti-fundamental} \\ a }} \Big[ v_a(t) - \big( \tilde \Delta_{a} - \pi \big) \Big] - \sum_{\substack{\text{fundamental} \\ a }} \Big[ v_a(t) + \big( \Delta_{a} - \pi \big) \Big] \right\} \, ,
\end{align}
while the second and the third lines give
\begin{align}\label{fund linear}
 & \frac{\left( \tilde n_a - n_a \right)}{4} N^{2} \int dt\, \rho(t)\, t^2 \nn \\
 & - \frac{i}{2} N^{3/2} \int dt\, \rho(t)\, t\, \left\{ \sum_{\substack{\text{anti-fundamental} \\ a }} \Big[ v_a(t) - \big( \tilde \Delta_{a} - \pi \big) \Big] - \sum_{\substack{\text{fundamental} \\ a }} \Big[ v_a(t) + \big( \Delta_{a} - \pi \big) \Big] \right\} \, . %
\end{align}
Combining Eq.\,\eqref{fund Li2} and Eq.\,\eqref{fund linear}, we obtain
\begin{align}\label{fundcontr}
 \cV^{\text{(anti-)fund}} & = - \frac{\left( \tilde n_a - n_a \right)}{4} N^{2} \int dt\, \rho(t)\, t\, |t| \nn \\
 & + \frac{i}{2} N^{3/2} \sum_{\substack{\text{anti-fundamental} \\ a }} \int dt\, \rho(t)\, |t|\, \Big[ v_a(t) - \big( \tilde \Delta_{a} - \pi \big) \Big] \nn \\
 & - \frac{i}{2} N^{3/2} \sum_{\substack{\text{fundamental} \\ a }} \int dt\, \rho(t)\, |t|\, \Big[ v_a(t) + \big( \Delta_{a} - \pi \big) \Big] \, .
\end{align}
Summing over nodes the first term vanishes, demanding that
\be
\sum_{a=1}^{|G|} \left( \tilde n_a - n_a \right) = 0 \, .
\ee
We see that we need to consider quivers where the {\it total} number of fundamentals equal the {\it total} number of anti-fundamentals.
For each single node this number can be different. 

\subsection{The index at large $N$}
\label{entropy rules N32}

We are interested in the large $N$ limit of the logarithm of the twisted partition function.

\subsubsection{Gauge vector contribution}
\label{entropy CS N32}

Given the expression for the matrix model  in Section \ref{index}, the Vandermonde determinant contributes to the logarithm of the index  as
\begin{align}\label{gauge entropy}
 \log \prod_{i \neq j} \left( 1- \frac{x_i^{(a)}}{x_j^{(a)}} \right) & =
 \log \prod_{i < j} \left( 1- \frac{x_j^{(a)}}{x_i^{(a)}} \right)^2 \left( - \frac{x_i^{(a)}}{x_j^{(a)}} \right) \nn \\
 & = i \sum_{i<j}^N \left( u_i^{(a)} - u_j^{(a)} + \pi \right)
 - 2 \sum_{i<j}^N \Li_{1} \left(e^{i \left( u_j^{(a)} - u_i^{(a)}\right)}\right) \, .
\end{align}
The first term is of $\cO(N^2)$ and, therefore, a source of the long-range forces and will be  canceled by the contribution coming from the chiral multiplets.
The second term is treated as in Appendix \ref{Bethe bi-fundamentals N32}, and gives
\be
 \re \log Z^{\text{gauge}} = - \frac{\pi^2}{3} N^{3/2}  \int dt\, \rho(t)^2 + \cO(N)  \, .
\ee

\subsubsection{Topological symmetry contribution}
\label{entropy topological N32}

A $\U(1)_a$ topological symmetry with flux $\ft_a$ contributes  as
\be\label{top entropy}
i \sum_{i=1}^{N} u_i^{(a)} \ft_a \, .
\ee
In the continuum limit, we get
\be
 \re \log Z^{\text{top}} = - \ft_a N^{3/2} \int dt\, \rho(t)\, t + \cO(N)  \, .
\ee

\subsubsection{Bi-fundamental contribution}
\label{entropy bi-fundamentals N32}

We can rewrite the contribution  to the twisted index of a bi-fundamental chiral multiplet transforming in the $(\overline{\bf N}, {\bf N})$ of $\U(N)_a \times \U(N)_b$, with magnetic flux $\fn_{(b,a)}$ and chemical potential $\Delta_{(b,a)}$  as:\footnote{The phases can be neglected, as we will be interested in $\log |Z|$.}
\begin{align}
 & \prod_{i=1}^{N} \left( \frac{x_i^{(b)}}{x_i^{(a)}} \right)^{\frac{1}{2} \left( \fn_{(b,a)} - 1\right)}
 \left( 1 - y_{(b,a)} \frac{x_i^{(b)}}{x_i^{(a)}} \right)^{\fn_{(b,a)} - 1} \times \nn \\
 & \prod_{i<j}^{N} (-1)^{\fn_{(a,b)} - 1} \left( \frac{x_i^{(a)} x_i^{(b)}}{x_j^{(a)} x_j^{(b)}} \right)^{\frac{1}{2} \left(\fn_{(b,a)} - 1\right)}
 \left( 1 - y_{(b,a)} \frac{x_j^{(b)}}{x_i^{(a)}} \right)^{\fn_{(b,a)}-1} \left( 1 - y_{(b,a)}^{-1} \frac{x_j^{(a)}}{x_i^{(b)}} \right)^{\fn_{(b,a)}-1} \, .
\end{align}
The first term in $\prod_{i}$ is subleading and the second term only contributes in the tail where $\delta v \approx - \Delta_{(b,a)}$,
\begin{align}
 & N \left(\fn_{(b,a)} - 1\right) \int dt\, \rho(t) \log \left( 1-e^{i\left(\delta v + \Delta_{(b,a)}\right)} \right) \nn \\
 & = - N^{3/2} \; \left(\fn_{(b,a)} - 1\right) \int_{\delta v \,\approx\, - \Delta_{(b,a)}} \hspace{-3em} dt\, \rho(t) \, Y_{(b,a)}(t) + \cO(N)
\end{align}
The first two terms in $\prod_{i<j}$ give a long-range force contribution to the index
\begin{equation}\label{forces chiral ba}
 \frac{i}{2} \left( \fn_{(b,a)} - 1 \right) \sum_{i<j} \left[ \left( u_i^{(a)} - u_j^{(a)} + \pi \right) + \left( u_i^{(b)} - u_j^{(b)} + \pi \right) \right] \, ,
\end{equation}
while the last two terms result in
\begin{align}
&  - N^{3/2} \left(\fn_{(b,a)}-1\right) \int dt\, \rho(t)^2 \left[ \Li_2\left( e^{i\left(\delta v+\Delta_{(b,a)}\right)} \right) + \Li_2 \left( e^{-i \left(\delta v + \Delta_{(b,a)}\right)} \right) \right] + \cO(N) \nn \\
 & = - N^{3/2} \left(\fn_{(b,a)}-1\right) \int dt\, \rho(t)^2 g_+' \left( \delta v(t) + \Delta_{(b,a)} \right) + \cO(N) \, .
\end{align}
A bi-fundamental field transforming in the $({\bf N}, \overline{\bf N})$ of $\U(N)_a \times \U(N)_b$, with magnetic flux $\fn_{(a,b)}$ and chemical potential $\Delta_{(a,b)}$ gives
the same contribution with the replacement $a\leftrightarrow b$ and $\delta v \rightarrow -\delta v$. 

The long-range force contribution of bi-fundamental fields at node $a$ cancels with the gauge contribution  in \eqref{gauge entropy},  provided that
\be
2 + \sum_{I \in a} (\fn_I - 1) = 0 \, ,
\ee
where the sum is taken over all chiral bi-fundamentals $I$ with an endpoint in $a$.

In picking the residues, we need to insert a Jacobian in the twisted index and evaluate everything else at the pole.
The matrix $\bB$ appearing in the Jacobian is $2 N \times 2 N$ with block form
\be
\label{Jacobian general}
\bB = \frac{\partial \big( e^{i B_j^{(a)}}, e^{i B_j^{(b)}} \big) }{ \partial( \log x_l^{(a)}, \log x_l^{(b)})} =
\mat{ x_l^{(a)} \dfrac{\partial e^{iB_j^{(a)}}}{\partial x_l^{(a)}} & x_l^{(b)} \dfrac{ \partial e^{i B_j^{(a)}}}{\partial x_l^{(b)}}  \\[1em]
x_l^{(a)} \dfrac{\partial e^{i B_j^{(b)}}}{\partial x_l^{(a)}} & x_l^{(b)} \dfrac{ \partial e^{i B_j^{(b)}}}{\partial x_l^{(b)}} }_{2 N \times 2 N} \, ,
\ee
and only contributes in the tails regions,\footnote{We refer the reader to \cite{Benini:2015eyy} for a detailed analysis of the Jacobian at large $N$.}%
\bea
-\log \det \bB = - N^{3/2} \sum_{\substack{\text{bi-fundamentals} \\ (b,a) \text{ and } (a,b)}}
 & \int_{\delta v \,\approx\, - \Delta_{(b,a)}} \hspace{-3em} dt\, \rho(t) \, Y_{(b,a)}(t) \\
 & + \int_{\delta v \,\approx\,\Delta_{(a,b)}} \hspace{-2em} dt\, \rho(t) \, Y_{(a,b)}(t) + \cO(N\log N) \, .
\eea

Summarizing, pairs of bi-fundamental fields contribute to the logarithm of the index  as
\begin{align}\label{indexbi}
 \re \log Z^{\text{bi-fund}}_{\text{bulk}} = - N^{3/2} \sum_{\substack{\text{bi-fundamentals} \\ (b,a) \text{ and } (a,b)}}
 & \int dt\, \rho(t)^2 \big[(\fn_{(b,a)}-1)\, g'_+ \left(\delta v(t) + \Delta_{(b,a)}\right) \nn \\
 & + (\fn_{(a,b)}-1)\, g'_- \left(\delta v(t) - \Delta_{(a,b)}\right)\big]\, .
\end{align}
The tails contribution is also given by
\begin{align}
 \re \log Z^{\text{bi-fund}}_{\text{talis}} = - N^{3/2} \sum_{\substack{\text{bi-fundamentals} \\ (b,a) \text{ and } (a,b)}}
 & \fn_{(b,a)}  \int_{\delta v \,\approx\, - \Delta_{(b,a)}} \hspace{-3em} dt\, \rho(t) \, Y_{(b,a)}(t) \nn \\
 & + \fn_{(a,b)} \int_{\delta v \,\approx\,\Delta_{(a,b)}} \hspace{-2em} dt\, \rho(t) \, Y_{(a,b)}(t) \, .
\end{align}

A field transforming in the adjoint of the $a$th gauge group with magnetic flux $\fn_{(a,a)}$ and chemical potential $\Delta_{(a,a)}$ only contributes to the bulk index.
To find its contribution we need to include an explicit factor of $1/2$ in the  expression (\ref{indexbi}) and take
\begin{equation}
v_b = v_a \, , \qquad \qquad  \Delta_{(b,a)} =  \Delta_{(a,b)} =  \Delta_{(a,a)} \, , \qquad \qquad  \fn_{(b,a)} =  \fn_{(a,b)} =  \fn_{(a,a)} \, .
\end{equation}

\subsubsection{Fundamental and anti-fundamental contribution}
\label{entropy fundamentals N32}

The fundamental and anti-fundamental fields contribute to the logarithm of the index  as
\begin{align}
 & \log \prod_{i=1}^{N} \prod_{\substack{\text{anti-fundamental} \\ a }} \left(x_i^{(a)}\right)^{\frac12 \left(\tilde \fn_a - 1\right)}
 \left[ 1 - \tilde y_a \left( x_i^{(a)} \right)^{-1} \right]^{\tilde \fn_a - 1} \times \nn \\
 & \prod_{\substack{\text{fundamental} \\ a }} \left(x_i^{(a)}\right)^{\frac12 \left( \fn_a - 1\right)}
 \left[ 1 - y_a^{-1} \left( x_i^{(a)} \right)^{-1} \right]^{\fn_a - 1} \, .
\end{align}
Using the scaling ansatz \eqref{ansatz alpha}, in the continuum limit we get
\begin{align}
 & \log \prod_{i=1}^N \prod_{\substack{\text{anti-fundamental} \\ a }} \left(x_i^{(a)}\right)^{\frac12 \left(\tilde \fn_a - 1\right)}
 \prod_{\substack{\text{fundamental} \\ a }} \left(x_i^{(a)}\right)^{\frac12 \left( \fn_a - 1\right)} \nn \\
 & = - \frac{1}{2} N^{3/2} \left[\sum_{\substack{\text{anti-fundamental} \\ a }} \left(\tilde \fn_a - 1\right)
 + \sum_{\substack{\text{fundamental} \\ a }} \left( \fn_a - 1\right) \right] \int dt\, \rho(t)\, t
 + \mathcal{O}(N) \, ,
\end{align}
and
\begin{align}
 &\log \prod_{i=1}^{N} \prod_{\substack{\text{anti-fundamental} \\ a }} \left[ 1 - \tilde y_a \left( x_i^{(a)} \right)^{-1} \right]^{\tilde \fn_a - 1}
 \prod_{\substack{\text{fundamental} \\ a }} \left[ 1 - y_a^{-1} \left( x_i^{(a)} \right)^{-1} \right]^{\fn_a - 1} \nn \\
 & = N^{3/2} \left[\sum_{\substack{\text{anti-fundamental} \\ a }} \left(\tilde \fn_a - 1\right)
 + \sum_{\substack{\text{fundamental} \\ a }} \left( \fn_a - 1\right) \right]
 \int_{t>0} dt\, \rho(t)\, t + \mathcal{O}(N) \, .
\end{align}
Putting the above equations together we find:
\begin{equation}
 \re \log Z^{\text{(anti-)fund}} = \frac{1}{2} N^{3/2} \left[\sum_{\substack{\text{anti-fundamental} \\ a }} \left(\tilde \fn_a - 1\right)
 + \sum_{\substack{\text{fundamental} \\ a }} \left( \fn_a - 1\right) \right] \int dt\, \rho(t)\, |t| \, .
\end{equation}

\section{An explicit example: the SPP theory}\label{SPP}

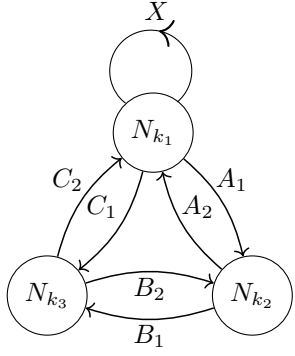
\begin{figure}[t]
\centering
\begin{tikzpicture}[font=\footnotesize, scale=0.9]
\begin{scope}[auto,
  every node/.style={draw, minimum size=0.5cm}, node distance=2cm];
\node[circle]  (UN)  at (0.35,1.7) {$N_{k_1}$};
\node[circle]  (UN2)  at (-1.2,-.7) {$N_{k_3}$};
\node[circle]  (UN3)  at (1.8,-.7) {$N_{k_2}$};
\end{scope}\draw[decoration={markings, mark=at position 0.45 with {\arrow[scale=2.0]{>}}}, postaction={decorate}, shorten >=0.7pt] (0.7,2.15) arc (-50:235:0.6cm);
\draw[draw=black,solid,line width=0.2mm,<-]  (UN) to[bend right=18] (UN2) ;
\draw[draw=black,solid,line width=0.2mm,->]  (UN) to[bend left=18]  (UN2) ;
\draw[draw=black,solid,line width=0.2mm,<-]  (UN) to[bend right=18] (UN3) ; 
\draw[draw=black,solid,line width=0.2mm,->]  (UN) to[bend left=18]  (UN3) ;  
\draw[draw=black,solid,line width=0.2mm,<-]  (UN2) to[bend right=18](UN3) ;
\draw[draw=black,solid,line width=0.2mm,->]  (UN2) to[bend left=18] (UN3) ;
\node at (1.5,1.0) {$A_1$};
\node at (1.0,0.6) {$A_2$};
\node at (0.3,-1.3) {$B_1$};
\node at (0.3,-0.6) {$B_2$};
\node at (-0.4,0.6) {$C_1$};
\node at (-0.9,1.0) {$C_2$};
\node at (0.4,3.5) {$X$};
\end{tikzpicture}
\caption{The SPP Chern-Simons-matter quiver.}
\label{SPP_quiver}
\end{figure}

We now consider, as an example, the quiver gauge theory which describes the dynamics of $N$ M2-branes at the suspended pinch point (SPP) singularity (see Fig.\,\ref{SPP_quiver}).
The  Chern-Simons levels are $(2k,-k,-k)$ and the superpotential coupling is given by
\begin{equation}\label{superpotential}
 W = \Tr\left[X \left( A_1 A_2 - C_1 C_2 \right) - A_2 A_1 B_1 B_2 + C_2 C_1 B_2 B_1 \right]\, .
\end{equation}
The marginality condition on the superpotential \eqref{superpotential0} impose  constraints on the chemical potential of the various fields
\begin{align}
 &\Delta_{A}+\Delta_{B}= \pi \, ,\qquad\qquad
 \Delta_{B}+\Delta_{C}= \pi \, ,\qquad\qquad
 2 \Delta_{A}+\Delta_{X}= 2 \pi \, ,
\end{align}
where we have used the symmetry of the quiver to set $\Delta_{A_1} = \Delta_{A_2} = \Delta_{A}$, and so on.
Hence,
\begin{equation}
 \Delta_{B} = \Delta \, , \qquad \qquad \Delta_X = 2 \Delta \, , \qquad \qquad
 \Delta_A = \Delta_C = \pi - \Delta \, ,
\end{equation}
and
\begin{align}\label{flux constraint}
 \fn_{B} = \fn \, , \qquad \qquad \fn_X = 2 \fn \, , \qquad \qquad
 \fn_A = \fn_C = 1 - \fn \, ,
\end{align}
where $\fn_I$ denotes the flavor magnetic flux of the field $I$. We assume $0\leq \Delta \leq 2\pi$ and we enforced condition \eqref{superpotential0}. One can check 
that all other solutions are related to the one we are presenting by a discrete symmetry of the quiver.\footnote{There is a solution for
\begin{align}
 &\Delta_{A}+\Delta_{B}= 3 \pi \, ,\qquad\qquad
 \Delta_{B}+\Delta_{C}= 3 \pi \, ,\qquad\qquad
 2 \Delta_{A}+\Delta_{X}= 4 \pi \, .
\end{align}
which is obtained, using the invariance of $Z$ under $y_I \to 1/y_I$, from \eqref{solution sum 2pi -- init}-\eqref{solution sum 2pi -- end} by performing the substitutions
\be
\mu \to - \mu \, ,\qquad k \to - k \, ,\qquad \Delta \to \pi - \Delta \, ,\qquad Y^{\pm} \to - Y^{\pm} \, .
\ee}

\subsection{The BAEs at large $N$}

The theory under consideration is invariant under
\begin{equation}
 A \leftrightarrow C \, , \qquad \qquad \U(N)^{(2)} \leftrightarrow \U(N)^{(3)} \, .
\end{equation}
Let us assume that the saddle-point solution is also invariant under this $\bZ_2$ symmetry. Thus, we can choose
\begin{equation}
 v_i^{(1)} = v_i \, , \qquad \qquad v_i^{(2)} = v_i^{(3)} = w_i \, .
\end{equation}
Given the rules of Section \ref{large N Bethe potential rules}, the Bethe potential reads
\begin{align}\label{large N Bethe potential}
 \frac{\mathcal{V}}{i N^{3/2}}&
 = 2 k \int dt\, t\, \rho(t)\, \delta v(t)
 + \int dt\, \rho(t)^2\, \Delta \left[ (\pi -\Delta ) (2 \pi -\Delta ) - 2 \delta v^2 \right] \nn \\
 & - \mu \left[\int dt\, \rho(t)-1\right]
 - \frac{2 i}{N^{1/2}}\int dt\, \rho(t)\, \left[\pm \Li_2 \left(e^{i\left[\delta v(t)\pm \left( \pi - \Delta\right) \right]}\right)\right] \, ,
\end{align}
where we defined
\begin{equation}
 \delta v(t) = w(t) - v(t) \, ,
\end{equation}
and we included the subleading terms giving rise to the equation of motion \eqref{tails}.\footnote{Notice that these terms are subleading in the equation of motion for $\rho$,
since $\Li_2$ is finite when its argument is one, while affect the equation of motion for $\delta v$ since $\Li_1$ is singular.}
The eigenvalue density distribution $\rho(t)$, which maximizes the Bethe potential, is a piece-wise function supported on $[t_{\ll}, t_{\gg}]$.
We define the inner interval as
\be
t_< \text{ s.t. } \delta v(t_<) = - \left( \pi - \Delta \right) \;,\qquad\qquad t_> \text{ s.t. } \delta v(t_>) = \pi - \Delta \, .
\ee
Schematically, we have:
\begin{center}
\begin{tikzpicture}[scale=2]
\draw (-2.,0) -- (2.,0);
\draw (-2,-.05) -- (-2, .05); \draw (-0.7,-.05) -- (-0.7, .05); \draw (0.7,-.05) -- (0.7, .05); \draw (2,-.05) -- (2, .05);
\node [below] at (-2,0) {$t_\ll$}; \node [below] at (-2,-.3) {$\rho=0$};
\node [below] at (-0.7,0) {$t_<$}; \node [below] at (-0.7,-.3) {$\delta v = - \left( \pi - \Delta \right)$}; \node [below] at (-0.7,-.6) {$Y^{-} = 0$};
\node [below] at (0.7,0) {$t_>$}; \node [below] at (0.7,-.3) {$\delta v = \pi - \Delta$}; \node [below] at (0.7,-.6) {$Y^{+} = 0$};
\node [below] at (2,0) {$t_\gg$}; \node [below] at (2,-.3) {$\rho=0$};
\end{tikzpicture}
\end{center}
The transition points are at
\be
\label{solution sum 2pi -- init}
t_\ll = -\frac{\mu }{2 k (\pi -\Delta )} \, ,\qquad t_< = - \frac{\mu }{k (2 \pi -\Delta )} \, ,\qquad t_> = \frac{\mu }{k (2 \pi -\Delta )} \, ,\qquad t_\gg = \frac{\mu }{2 k (\pi -\Delta )} \, .
\ee
In the \emph{left tail} we have
\be
\begin{aligned}
\rho & = \frac{1}{2 \Delta^2} \left(\frac{\mu }{\pi -\Delta }+2 k t \right)\, , \qquad\quad
\delta v= - \left( \pi - \Delta \right) \\[.5em]
Y^{-} & = -\frac{ \mu + k (2 \pi - \Delta ) t }{\Delta }
\end{aligned}
\qquad\qquad t_\ll < t < t_< \, .
\ee
In the \emph{inner interval} we have
\be
\begin{aligned}
\rho & = \frac{\mu }{2 (\pi -\Delta ) (2 \pi -\Delta ) \Delta } \, ,\qquad\quad
\delta v = \frac{k (\pi -\Delta ) (2 \pi -\Delta ) t}{\mu }
\end{aligned}
\qquad\qquad t_< < t < t_>
\ee
and $\delta v'>0$. In the \emph{right tail} we have
\be
\begin{aligned}
\rho & = \frac{1}{2 \Delta^2} \left( \frac{\mu }{\pi -\Delta }-2 k t \right) \, ,\qquad\quad
\delta v = \pi - \Delta \\[.5em]
Y^{+} & = -\frac{\mu - k (2 \pi - \Delta ) t}{\Delta }
\end{aligned}
\qquad\qquad t_> < t < t_\gg \, .
\ee
Finally, the normalization fixes
\be
\label{solution sum 2pi -- end}
\mu = 2  k^{1/2} (\pi -\Delta ) (2 \pi -\Delta ) \sqrt{\frac{\Delta}{4 \pi -3 \Delta }} \, .
\ee
$\mu >0$ implies the following inequality
\begin{equation}
 0 < \Delta < \pi \, .
\end{equation}
For $k>1$  there can be discrete $\mathbb{Z}_k$ identifications among the chemical potential which can affect the final result.
We have not been too careful about them here.

\subsection{The index at large $N$}

The rules of the large $N$ twisted index imply that the free energy functional is
\begin{equation}\label{Z large N functional}
  \begin{aligned}
  \quad\rule{0pt}{2em} \re \log Z &= - N^{3/2} \int dt\, \rho(t)^2 \left[ \Delta  (4 \pi -3 \Delta )+\fn \left(3 \Delta^2-6 \pi  \Delta + 2 \pi^2 - 2 \delta v^2 \right) \right] \quad \\
  \rule[-2em]{0pt}{1em} &\quad - N^{3/2} 2 ( 1- \fn) \int_{\delta v \,\approx\, - (\pi - \Delta)} \hspace{-2em} dt\, \rho(t) \, Y^{-}(t)
  - N^{3/2} 2 ( 1- \fn) \int_{\delta v \,\approx\, (\pi - \Delta)} \hspace{-2em} dt\, \rho(t) \, Y^{+}(t) \, .
 \end{aligned}
\end{equation}
We should take the solution to the BAEs, plug it back into the functional \eqref{Z large N functional} and compute the integral.
Doing so, we obtain the following expression for the logarithm of the index:
\begin{align}
 \re \log Z & = -\frac43 \frac{k^{1/2} N^{3/2}}{(4 \pi -3 \Delta )^{3/2} \sqrt{\Delta }} \times \nn \\
 & \left[\Delta \left(7 \Delta^2-18 \pi  \Delta +12 \pi^2\right) +
 \fn \left(-6 \Delta^3+19 \pi \Delta^2-18 \pi^2 \Delta +4 \pi^3\right) \right] \, .
\end{align}
Notice that
\begin{equation}
 {\quad \rule[-1.4em]{0pt}{3.4em}
 \re\log Z = - \frac{2}{\pi} \, \wb{\mathcal{V}}(\Delta) \,
 - \left[ \left(\fn - \frac{\Delta}{\pi}\right) \frac{\partial \wb{\mathcal{V}}(\Delta)}{\partial \Delta}
  \right] \, ,
 \quad}
\end{equation}
where
\begin{equation}
 \wb{\mathcal{V}}(\Delta) = \frac{4}{3} k^{1/2} N^{3/2} (\pi -\Delta ) (2 \pi -\Delta ) \sqrt{\frac{\Delta }{4 \pi -3 \Delta }}  \, ,
\end{equation}
as expected from the index theorem.

\section{Derivation of general rules for theories with $N^{5/3}$ scaling of the index}
\label{general rules N53}

We assume that in the large $N$ limit the eigenvalues corresponding to all the gauge groups are the same to leading order in $N$ and they behave as
\begin{equation}\label{ansatz N53}
 u^{(a)} (t) = N^{\alpha} (i t + v(t))\, ,
\end{equation}
for some $0<\alpha<1$.
We also assume that $\sum_{a=1}^{|G|} k_a \neq 0$ in the following discussion, as appropriate for theories with a massive type IIA behavior.
The long-range force analysis is identical to Appendix \ref{general rules N32}. Here we discuss the $N^{5/3}$ contributions.

\subsection{Bethe potential at large $N$}
\label{Bethe potential rules N53}

Each group $a$ with CS level $k_a$ contributes to the finite $N$ Bethe potential as
\begin{equation}
 \cV^{\text{CS}} = - \frac {k_a}{2} \sum_{i=1}^N \left(u_i^{(a)}\right)^2 \, .
\end{equation}
Using the scaling ansatz \eqref{ansatz N53}, we find
\begin{equation}
 - i k_a N^{2 \alpha + 1} \int dt\, \rho(t)\, t\, v(t) + \frac{k_a}{2} N^{2 \alpha + 1} \int dt\, \rho(t)\, \left(t^2 - v(t)^2\right)\, .
\end{equation}

To obtain the large $N$ behavior of a bi-fundamental field between $\U(N)_a \times \U(N)_b$ we follow the same strategy as in Section \ref{Bethe bi-fundamentals N32}.
For example, consider
\be
\sum_{i<j}^N \Li_2 \left(e^{i\left(u_j^{(b)} - u_i^{(a)} +\Delta_{I}\right)}\right) \, .
\ee
We first write the dilogarithm function as a power series, \ie,
\be
\Li_2(e^{iu}) = \sum_{k=1}^\infty \frac{e^{iku}}{k^2} \, ,
\ee
and then consider the integral
\begin{align}
 I_k = \int_t dt' \rho(t') e^{i\left(u_b(t') - u_a(t) +\Delta_{I}\right)}
 = \int_t e^{-k N^\alpha (t' - t)} \sum_{j=0}^{\infty} \frac{(t'-t)^j}{j!} \partial_x^j \left[\rho(x) e^{i k \left[N^\alpha \left(v(x) - v(t)\right) + \Delta_{I}\right]}\right]_{x=t}\, .
\end{align}
Performing the integral in $t'$ we find
\begin{align}
 \int_t dt' e^{-k N^\alpha (t' - t)} \frac{(t'-t)^j}{j!} = \left(k N^\alpha\right)^{-j-1}\, .
\end{align}
Next, we extract the leading contribution of the left over term, \ie,
\begin{align}
 \partial_x^j \left[\rho(x) e^{i k \left[N^\alpha \left(v(x) - v(t)\right) + \Delta_{I}\right]}\right]_{x=t} & \sim
 \left(i k N^\alpha\right)^j \left[v'(x)^j \rho(x) e^{i k \left[N^\alpha \left(v(x) - v(t)\right) + \Delta_{I}\right]}\right]_{x=t}\nn \\&
 =\left(i k N^\alpha\right)^j v'(t)^j \rho(t) e^{i k \Delta_{I}}\, .
\end{align}
Bringing the pieces together we find
\begin{align}
 I_k = \frac{e^{i k \Delta_{I}}}{k} \rho(t) N^{-\alpha} \sum_{j=0}^{\infty} \left[i v'(t)\right]^j
 = \frac{e^{i k \Delta_{I}}}{k} N^{-\alpha} \frac{\rho(t)}{1- i v'(t)}\, .
\end{align}
Thus,
\begin{align}
 \sum_{i<j}^N \Li_2 \left(e^{i\left(u_j^{(b)} - u_i^{(a)} +\Delta_{I}\right)}\right) =
 N^{2-\alpha} \int dt \Li_3 \left(e^{i \Delta_{I}}\right) \frac{\rho(t)^2}{1 - i v'(t)} \, .
\end{align}
Following the same steps as before, we get
\begin{equation}
 \cV^{\text{bi-fund}} = i g_+\left(\Delta_{I}\right) N^{2 - \alpha} \int dt\, \frac{\rho(t)^2}{1-i v'(t)} \, .
\end{equation}
To have a nontrivial saddle-point, we need $\alpha=1/3$ which ensure that the Chern-Simons terms and the matter
contributions scale with the same power of $N$.

The contribution of (anti-)fundamental fields to the Bethe potential is given by [see  Section \ref{Bethe fundamentals N32}],
\be
 \cV^{\text{(anti-)fund}} = \frac{\left( \tilde n_a - n_a \right)}{4} N^{5/3}\int dt\, \rho(t)\, \sign(t) \left[ i t + v(t) \right]^2\, .
\ee
Notice that, when the {\it total} number of fundamental and anti-fundamental fields in the quiver are equal, this contribution vanishes.

\subsection{The index at large $N$}
\label{entropy rules N53}

The contribution of the Vandermonde determinant to the index can be found using the same techniques presented in Appendix \ref{general rules N32}.
We need to compute
\be
- 2 \sum_{i<j}^N \Li_{1} \left(e^{i \left( u_j^{(a)} - u_i^{(a)}\right)}\right) \, .
\ee
After a small calculation we find,
\be
 \re \log Z^{\text{gauge}} = -\frac{\pi^2}{3} N^{5/3} \int dt\, \frac{\rho(t)^2}{1-i v'(t)}\, .
\ee

We now consider a bi-fundamental field with chemical potential $\Delta_I$ and flavor magnetic flux $\fn_I$.
Using the same methods given in Appendix \ref{general rules N32}, we obtain
\begin{equation}
 \re \log Z^{\text{bi-fund}}_{\text{bulk}} = - (\fn_{I}-1)\, g'_+ \left(\Delta_{I}\right) N^{5/3} \int dt\, \frac{\rho(t)^2}{1-i v'(t)}\, .
\end{equation}

The contribution of (anti-)fundamental fields to the index, at large $N$, is subleading and they just contribute through the Bethe potential.
\section{Polylogarithms}

In this appendix we review the Polylogarithms and their properties which we used in the paper.
The polylogarithm function $\Li_n(z)$ is defined by a power series
\be
\Li_n(z) = \sum_{k=1}^\infty \frac{z^k}{k^n} \, ,
\ee
in the complex plane over the open unit disk, and by analytic continuation outside the disk.
For $z = 1$ the polylogarithm reduces to the Riemann zeta function
\be
\Li_n(1) = \zeta(n)\, , \qquad \text{for} \quad \re n >1 \, .
\ee
The polylogarithm for $n=0$ and $n=1$ is
\be
\Li_0(z) = \frac{z}{1-z}\, , \qquad \qquad  \Li_1(z) = - \log(1-z) \, .
\ee
Notice that $\Li_0(z)$ and $\Li_1(z)$ diverge at $z=1$.
For $n\geq 1$, the functions have a branch point at $z=1$ and we shall take the principal determination with a cut $[1,+\infty)$ along the real axis.
The polylogarithms fulfill the following relations
\be
\partial_u \Li_n(e^{iu}) = i\, \Li_{n-1}(e^{iu}) \;,\qquad\qquad \Li_n(e^{iu}) = i \int_{+i\infty}^u \Li_{n-1}(e^{iu'})\, du' \;.
\ee
The functions $\Li_n(e^{iu})$ are periodic under $u \to u+2\pi$ and have branch cut discontinuities along the vertical line $[0, -i\infty)$ and its images.
For $0< \re u < 2\pi$, polylogarithms satisfy the following inversion formul\ae{}\footnote{The inversion formul\ae{} in the domain $-2\pi < \re u < 0$ are obtaind by sending $u \to -u$.}
\bea
\label{reflection formulae}
\Li_0(e^{iu}) + \Li_0(e^{-iu}) &= -1 \\
\Li_1(e^{iu}) - \Li_1(e^{-iu}) &= -iu + i\pi \\
\Li_2(e^{iu}) + \Li_2(e^{-iu}) &= \frac{u^2}2 - \pi u + \frac{\pi^2}3 \\
\Li_3(e^{iu}) - \Li_3(e^{-iu}) &= \frac i6 u^3 - \frac{i\pi}2 u^2 + \frac{i\pi^2}3 u \, .
\eea
One can find the formul\ae{} in the other regions by periodicity.

%

\providecommand{\href}[2]{#2}\begingroup\raggedright\endgroup

\end{document}